\documentclass{llncs}
\makeatletter
\let\vec\@undefined
\expandafter\let\csname vec \endcsname\@undefined
\DeclareMathAccent{\vec}{\mathord}{letters}{"7E}
\makeatother
\usepackage{amsmath,amssymb,mathtools}
\usepackage[hidelinks]{hyperref}
\usepackage[nameinlink,capitalise]{cleveref}
\usepackage{xcolor}

\spnewtheorem{fact}[theorem]{Fact}{\bfseries}{\itshape}
\crefname{fact}{Fact}{Facts}
\Crefname{fact}{Fact}{Facts}

\newcommand{\F}{\mathbb F}
\newcommand{\E}{\mathbb E}
\newcommand{\C}{\mathcal C}
\newcommand{\Pau}{\mathcal P}
\newcommand{\Tr}{\operatorname{Tr}}

\newcommand{\Adv}{\operatorname{Adv}}
\newcommand{\one}{\mathbf 1}
\newcommand{\ip}[2]{\langle #1,#2\rangle}

\title{Quantum Leakage Resilience of Shamir Secret Sharing}
\titlerunning{Quantum Leakage Resilience of Shamir Secret Sharing}
\author{Rishabh Batra\inst{1} \and
Fuyuki Kitagawa\inst{2} \and
Ryo Nishimaki\inst{2} \and
Takashi Yamakawa\inst{2}}
\authorrunning{R. Batra et al.}
\institute{EPFL \and NTT Social Informatics Laboratories}

\begin{document}
\maketitle
\thispagestyle{plain}

\begin{abstract}
We initiate the study of quantum leakage resilience of unmodified Shamir secret sharing over prime fields. A well-studied leakage model for Shamir's secret sharing classically is single-bit local leakage from each share. We consider its quantum analogue where for each party, a local leakage channel takes as input the party's share and outputs a leaked qubit.  Without
preshared entanglement, we show that the distinguishing advantage is $2^{-\Omega(n)}$ when the threshold rate
$t/n=\tau$ exceeds $\tau_\star\approx0.73339$ by a fixed positive margin.
The proof uses finite-field Fourier analysis and properties of MDS codes. 

More generally, we allow disjoint entangled blocks of any fixed maximum size where there is no entanglement between different blocks or with the adversary, and
each block emits at most a fixed number of qubits. Security holds when the threshold rate is high enough (sufficiently close to one). We then allow a specified set of devices to share
entanglement with the adversary. We show that security holds even when a linear number of devices ($\alpha n$ for small $\alpha>0$) share entanglement with each other and with the adversary for a large enough threshold rate.  

As a complementary negative result, we also show that even classical single-bit leakage makes Shamir scheme insecure if we allow arbitrarily large entanglement between the leakage devices. A GHZ state shared by exactly $t$ leakage devices
makes even classical one-bit leakage insecure, without any
entanglement with the adversary. In this attack, each
participating device emits only one classical bit, and their
joint parity distinguishes any chosen pair of secrets with a constant advantage. Thus, for fixed threshold rates above $\tau_\star$, the maximum number of devices
that may share arbitrary entanglement with one another and with the
adversary while preserving security is linear in $n$ up to constant
factors, although the optimal support fraction remains open.

\keywords{Shamir secret sharing \and leakage resilience \and quantum leakage
\and MDS codes}
\end{abstract}

\section{Introduction}

Secret sharing distributes a secret among several parties so that only authorized
sets of parties can reconstruct it, whereas unauthorized sets learn nothing about
the secret.  This reduces reliance on any single holder for security, with
applications in threshold cryptography, distributed key management, and secure
multiparty computation; see, for example, the survey of Beimel~\cite{Beimel2011}.
The standard security guarantee holds even when an adversary learns all shares held
by an unauthorized set, provided it obtains no information about any other share.

Shamir's polynomial-based construction is one of the most commonly used secret-sharing
schemes~\cite{Shamir1979}.  It realizes a $t$-out-of-$n$ threshold
access structure: every set of at least $t$ parties is authorized, and every
smaller set is unauthorized.  To share $s\in\F_p$, the scheme chooses a uniformly
random polynomial $f$ of degree less than $t$, subject to $f(0)=s$, and gives
party $i$ the field element $f(x_i)$ as the share.  Any $t$ evaluations determine $f(0)$
by interpolation, whereas every set of fewer than $t$ shares is statistically
independent of the secret.

Shamir's algorithms are very simple and efficient, which helps explain its
widespread use in practice.  The public evaluation points can be fixed beforehand, and the
Lagrange coefficients for a reconstruction set can also be precomputed.  More
importantly, the scheme is linear; in particular, the secret can be reconstructed by simply adding the shares together after multiplying them by specific constants.  These properties make Shamir sharing a natural building
block for larger protocols, including general secure multiparty computation based
on linear secret-sharing schemes~\cite{CramerDamgardMaurer2000}.

The standard security model has a sharp gap between fully exposed shares and
shares about which the adversary learns nothing.  Implementations of cryptographic schemes in practice may 
reveal partial information through timing or electromagnetic
measurements.  Leakage-resilient cryptography asks whether security of protocols hold against such
partial exposure of information to the adversary~\cite{MicaliReyzin2004,DziembowskiPietrzak2008}.  For secret
sharing, such leakage may come from multiple devices, including from the completely hidden
shares.  Ordinary threshold privacy gives no guarantee in this
situation.  Even one bit from each share may collectively reveal information about
the secret.  This inspired the study of the local-leakage model, which allows the adversary to choose a short-output function of each share and observe all the
outputs~\cite{BDIR2021}.  Proving security for an existing scheme is particularly
useful: it protects a familiar scheme against leakage without requiring a new sharing
algorithm or a change to protocols that already use it.

There has been substantial work on classical local leakage from secret sharing (see Related work).
With advances in quantum technologies, a natural question is whether the widely
used Shamir scheme is also resilient to quantum leakage.  We initiate the study of
quantum leakage resilience of Shamir's scheme.

In the basic model, each device receives its classical share and outputs one qubit, and the adversary may measure all outputs jointly. Classical one-bit leakage is a special case, but its security guarantees do not directly establish security for quantum outputs. Preshared entanglement raises a further question: how does the security of local leakage depend on the quantum entanglement available to the devices before they receive their shares?

Our first result establishes exponentially small distinguishing advantage for Shamir sharing over $\mathbb F_p$, for every prime $p>n$ and every choice of distinct nonzero evaluation points, when the devices have no preshared entanglement and the threshold rate exceeds $\tau_\star\approx0.73339$ by a fixed positive margin. We extend this to also show that security holds at sufficiently high threshold rates for when disjoint blocks of any fixed maximum size are allowed to share entanglement. In fact, this result permits arbitrary joint computation within each block with a separate fixed bound on its output qubits and no entanglement between blocks or with the adversary.

 We then allow a specified set of devices to share
entanglement with the adversary. We show that security holds even when a linear number of devices ($\alpha n$ for small $\alpha>0$) share entanglement with each other and with the adversary for a large enough threshold rate.  We also show that even classical single-bit leakage makes Shamir scheme insecure if we allow arbitrarily large entanglement between the leakage devices. A GHZ state shared by exactly $t$ leakage devices
makes even classical one-bit leakage insecure, without any
entanglement with the adversary.
Each
participating device emits only one classical bit, yet their
joint parity distinguishes any chosen pair of secrets with advantage
$\cos^2(\pi/(2p))=1-O(p^{-2})$. Thus, for every fixed threshold rate $\tau>\tau_\star$, when we allow devices to
share arbitrary entanglement with one another and with the adversary,
the maximum number of participating devices compatible with security
is linear in $n$ up to constant factors. The threshold proved by our
security bound is not claimed to be optimal.
\subsection{Our results}

The constant $\tau_\star$ used in the theorem statements is from \cref{cor:direct-product} and its value is approximately $0.73339$.   Let $p>n$ be an odd prime.

\begin{theorem}[Informal; unentangled quantum leakage]
  \label{thm:intro-product}
 Consider an $(n,t)$-Shamir secret sharing scheme over the finite-field $\F_p$ with threshold rate
$t/n=\tau$. We allow local leakage from each share, where we apply an arbitrary local channel to each share and
output one qubit.  If there is no prior entanglement and
$\tau>\tau_\star$ with a fixed positive gap, the advantage in distinguishing any two secrets is
$\operatorname{negl}(n)$, even when the adversary jointly measures all outputs.
\end{theorem}

\begin{theorem}[Informal; quantum leakage from disjoint entangled blocks, no entanglement with the adversary]
  \label{thm:intro-blocks}
Fix a constant $b$. Partition the leakage devices into disjoint blocks
of at most $b$ devices each. Devices may share entanglement within a block.
There is no shared quantum state between blocks or with the adversary.
A block containing $m$ shares emits an $m$-qubit state.
There exists a threshold $\tau_\star(b)<1$, depending only on $b$,
such that, for every fixed threshold rate $\tau>\tau_\star(b)$,
the distinguishing advantage is $\operatorname{negl}(n)$, even under a
joint measurement of all outputs. \Cref{tab:block-thresholds} lists
approximations to sufficient thresholds when all blocks have the same size.
\end{theorem}
The same argument as above allows separate fixed bounds $b$ on input shares and
$\ell$ on output qubits per block, with a sufficient threshold depending
on both (see \Cref{sec:bl}).

\begin{theorem}[Informal; quantum leakage with devices entangled with the adversary]
  \label{thm:intro-ent}
Consider an $(n,t)$-Shamir secret sharing scheme over the finite field $\F_p$ with threshold rate
$t/n=\tau$. We now let $r$ of the $n$ devices share arbitrary entanglement with one another
and with the adversary, while the remaining devices do unentangled local leakage.  Let $r/n=\delta<1$, with $1-\delta$ bounded below
by a positive constant.  Each of these $r$ devices can now use its part of the shared entanglement (with the adversary) along with the knowledge of its share to compute the leaked qubit.  The distinguishing
advantage is $\operatorname{negl}(n)$ if the following inequality holds with a fixed positive gap:
\begin{equation*}
  \frac{\tau-\delta}{1-\delta}>\tau_\star.
\end{equation*}
In particular, security holds for $r=o(n)$ whenever $\tau>\tau_\star$ with a fixed positive gap, and also for $r=\alpha n$ with sufficiently
small constant $\alpha>0$ for sufficiently large $\tau$ (see \Cref{tab:support-rates}).
\end{theorem}

\noindent\emph{{Remark:}} The security guarantees for quantum leakage resilience also extend to a combination of the two models. \textbf{The strongest leakage model that we handle is the following: a set of leakage devices of linear size may share arbitrary entanglement with the adversary while
the remaining devices may form blocks of constant size independent of each other and the adversary's entangled block.} Security follows whenever the
residual rate $(t-r)/(n-r)$ exceeds the corresponding block
threshold by a fixed positive margin.

\begin{theorem}[Informal; a classical one-bit-leakage attack with entanglement]
  \label{thm:intro-ghz}
Consider an $(n,t)$-Shamir secret sharing scheme over the finite-field $\F_p$. There is a classical single-bit leakage  attack using a GHZ state shared by $t$ devices such that when each device leaks just
one classical bit, the  distinguishing advantage for two secrets is at least
$\cos^2(\pi/(2p))=1-O(p^{-2})$.
\end{theorem}

Our results strengthen the security guarantees of unmodified Shamir sharing without changing its sharing or reconstruction algorithms. Together, the security bounds and the above attack show that quantum leakage resilience depends on both the size of the leaked outputs and the entanglement available to the leakage devices.

\subsection{Techniques}

Numerous works on classical leakage-resilient Shamir secret sharing use Fourier analysis techniques for proving security~\cite{BDIR2021,MajiEtAlCrypto2021,KleinKomargodski2023}. In the classical one-bit leakage model, each local leakage function
$f:\mathbb{F}_p\to\{0,1\}$ is Boolean. The identity $f^2=f$,
together with Parseval's identity and conjugate symmetry, gives
useful constraints on its scalar Fourier coefficients.

For quantum leakage, a one-qubit output can be described by three
real Pauli coefficient functions satisfying the joint constraint
$r_X(y)^2+r_Y(y)^2+r_Z(y)^2\leq1$.
These functions still satisfy Parseval's identity and conjugate
symmetry. Deriving suitable Fourier bounds from their joint constraint and using them to control the distinguishing advantage becomes more complicated than the classical single-bit leakage case. We use Pauli
orthogonality and then sharpen the analysis by bounding the trace
norms of matrix-valued Fourier coefficients directly. Proving leakage resilience against quantum adversaries that may also share entanglement with the leakage devices is an additional challenge that we need to take care of.

We adapt a matrix-valued Fourier analysis argument that lets us deal with quantum leakage. We write the qubit leaked by device
$i$ in its Bloch expansion,
\begin{equation*}
  \rho_i(y)=\frac12\bigl(I+r_{i,X}(y)X+r_{i,Y}(y)Y+r_{i,Z}(y)Z\bigr),
\end{equation*}
where $r_{i,X}(y),r_{i,Y}(y),r_{i,Z}(y)\in\mathbb{R}$ and
$r_{i,X}(y)^2+r_{i,Y}(y)^2+r_{i,Z}(y)^2\leq1$. Let
$r_{i,I}(y)=1$ and $
  W_i(a)=\sum_{P\in\{I,X,Y,Z\}}|\widehat r_{i,P}(a)|^2,$ 
where the hat denotes the Fourier transform over $\F_p$.
Parseval's identity and conjugate symmetry give
\begin{equation*}
  \sum_{a\in\F_p}W_i(a)\leq2,
  \qquad W_i(a)\leq\frac12\quad(a\neq0).
\end{equation*}
The shares for a secret $s$ are uniform on the affine coset
$\C_s=\C_0+s\one$, where $\C_0$ is the space of zero-secret sharings.
Using Pauli orthogonality, we bound the distinguishing advantage
by the following sum (similar to the idea in \cite{BDIR2021}):
\begin{equation*}
 \Adv(\rho_{s_0},\rho_{s_1})
  \leq \sum_{a\in\C_0^\perp\setminus\{0\}}
       \prod_{i=1}^n\sqrt{W_i(a_i)}.
\end{equation*}
We bound this sum using Cauchy-Schwarz and the elementary properties
of the dual MDS code. The code $\C_0^\perp$ has minimum distance $t$,
so every nonzero codeword has at most $n-t$ zero coordinates.
This limits the number of factors in the product that can exceed
$1/\sqrt{2}$: a nonzero coordinate contributes a factor at most
$1/\sqrt{2}$, whereas a zero coordinate contributes a factor at most
$\sqrt{2}$.
Combining this with Cauchy-Schwarz on suitable blocks
of coordinates gives a simple proof of security for threshold
rates above $5/6$.

For the stronger bound, we analyze independent blocks of leakage devices.
Let $B_1,\ldots,B_M$ be a partition of the devices into blocks of
sizes $m_j=|B_j|\leq b$. Each block may apply an arbitrary joint map
to its shares and output an $m_j$-qubit state
$\sigma_j(y_{B_j})$. There is no entanglement between different
blocks or with the adversary. We apply the Fourier transform directly to these density matrices and
write
$
  h_j(u)=\|\widehat\sigma_j(u)\|_1.$

As in the elementary argument, averaging over the randomness of the shares
leaves only frequencies in the dual code. Fourier inversion and the
triangle inequality then give (\Cref{lem:direct-block-fourier}):
\[
  \Adv(\rho_{s_0},\rho_{s_1})
  \leq
  \sum_{a\in\C_0^\perp\setminus\{0\}}
       \prod_{j=1}^M h_j(a_{B_j}).
\]
Thus we need to control a sum of products of local Fourier trace norms.

We obtain analogous bounds using Parseval's identity, trace-norm duality,
and elementary trigonometric estimates
(\Cref{lem:block-fourier-budget}):
\[
  \begin{gathered}
    h_j(0)=1,\qquad
    \sum_{u\neq0}h_j(u)^2\leq2^{m_j}-1,\\[4pt]
    h_j(u)\leq\kappa_{m_j}+\frac{\pi}{p}
    \qquad(u\neq0),
  \end{gathered}
\]
where
\[
  \kappa_m=\frac{2^m}{\pi}
  \sin\left(\frac{\pi}{2^m}\right)<1.
\]
 The block MDS
Brascamp-Lieb inequality (\Cref{lem:unequal-block-bl}) combines these local estimates into a bound
on the sum over dual codewords. We also use the minimum distance of the dual code as before. Every nonzero
codeword has at least $t$ nonzero coordinates and therefore meets
at least $\lceil t/b\rceil$ blocks. To incorporate this information,
we multiply each factor corresponding to a nonzero block frequency by a parameter $\lambda\geq1$.
Every product gains a factor of at least
$\lambda^{\lceil t/b\rceil}$, so we can divide the resulting bound
by this factor. Optimizing $\lambda$ gives exponential security
when the threshold rate is sufficiently large. 

For every
fixed maximum block size $b$, the sufficient threshold is strictly
less than one. The sharper thresholds when the blocks are of the same size are
listed in \Cref{tab:block-thresholds}.
Taking $b=1$ gives exponentially small distinguishing advantage
against unentangled one-qubit leakage whenever $t/n$ exceeds
$\tau_\star\approx0.73339$ by a fixed positive margin.

For devices in a specified set $E\subset[n]$, with $|E|=r<t$,
sharing entanglement with one another and with the adversary, we first give
the adversary the complete classical shares $Y_E$. From these shares and
the fixed leakage strategy, it can prepare the joint state of the supported
devices' outputs and its reference register. This channel is
independent of the secret, so data processing bounds the original
distinguishing advantage by that in the full-revelation experiment. Since $r<t$, the revealed shares $Y_E$ are uniform and independent
of the secret. Conditioned on $y_E$, the remaining shares
are coordinatewise affine images of a Shamir sharing with $n-r$
parties and threshold $t-r$. We can now apply the unentangled bound
on this smaller Shamir secret sharing scheme to get the desired bound.

Finally, we give an attack that adapts the GHZ phase-encoding and parity technique of
Buhrman et al.~\cite{BuhrmanVanDamHoyerTapp1999} to the Shamir scheme. For a reconstruction set $T$ of
size $t$, write $s=\sum_{i\in T}\lambda_i y_i$ and set $\omega=e^{2\pi i/p}$. Starting from a
$t$-qubit GHZ state, device $i$ applies the phase
$\omega^{\lambda_i y_i}$ to its $|1\rangle$ component. The resulting
state is
\[
  \frac{|0^t\rangle+\omega^s|1^t\rangle}{\sqrt2}.
\]
The product of the local Pauli-$X$ measurement outcomes,
represented as $\pm1$, has expectation $\cos(2\pi s/p)$. For the secrets
$0$ and $(p-1)/2$, this gives distinguishing advantage
$\cos^2(\pi/(2p))$. Rescaling the local phases and adding a fixed phase
at one device gives the same advantage for any chosen pair of distinct
secrets. Thus entanglement enables a classical-output attack even though
each device emits only one bit.
\subsection{Related work}

Benhamouda, Degwekar, Ishai, and Rabin initiated the study of local
leakage resilience for standard linear secret-sharing schemes and proved security
for Shamir sharing over large prime-order fields~\cite{BDIR2021}. Since then, a significant amount of research has been done on classical leakage-resilience of Shamir secret sharing. 
Maji et al. constructed new locally leakage-resilient linear secret-sharing
schemes and sharpened the earlier one-bit analysis for unmodified Shamir
sharing~\cite{MajiEtAlCrypto2021}.  Subsequent works lowered the sufficient
worst-case threshold to approximately $0.668n$~\cite{MajiEtAl2022,KleinKomargodski2023,Kasser2024}.  Nguyen proved
security at every fixed positive rate for almost every tuple of one-bit leakage
functions over sufficiently large prime fields, and gave a two-bit leakage attack
when $t=O(\sqrt n)$ and $p=\Theta(n)$~\cite{Nguyen2024}.  

In the quantum setting, leakage resilience has attracted considerable
interest. There have been a series of works that study unbounded classical leakage from
quantum encodings using tools from unclonable cryptography
\cite{CakanEtAl2023,CakanGoyal2024,CakanGoyal2025}.
Ananth, Kaleoglu, and Yuen construct secret-sharing schemes with quantum
shares that resist bounded classical leakage, using Haar-random states
and quantum state designs~\cite{AnanthKaleogluYuen2024}.

Boddu, Goyal, Jain, and Ribeiro construct classical-share schemes
resilient to bounded local quantum leakage, allowing preshared
entanglement and full exposure of an unauthorized set
\cite{BodduEtAl2025}.
Bergamaschi and Boddu also construct schemes for bounded joint quantum
leakage from one unauthorized set to another
\cite{BergamaschiBoddu2026}. These non-linear schemes are extractor-based and more complicated than the traditional Shamir scheme. Our goal is to understand the quantum leakage resilience of the
 widely used unmodified Shamir scheme. 

In a different line of work, Sun and Wootters connected local leakage resilience of secret sharing to
worst-case optimal polynomial intersection~\cite{SunWootters2026}.  Horinaga and
Yamakawa used an MDS Brascamp-Lieb inequality in this setting and gave a
worst-case quantum algorithm for optimal polynomial intersection
\cite{HorinagaYamakawa2026}. We combine a block version of their MDS Brascamp-Lieb inequality
with Fourier estimates for matrix-valued
functions to analyze quantum leakage from Shamir sharing.

\paragraph{Organization.}
\Cref{sec:prelim} introduces the preliminaries and notation.
\Cref{sec:reduction} combines Pauli expansions with Fourier analysis
to bound the distinguishing advantage by a weighted sum over the dual
Shamir code. \Cref{sec:cauchy} gives an elementary bound on this sum using
Cauchy-Schwarz and the MDS properties of the code. \Cref{sec:bl} obtains
security for disjoint entangled blocks using the MDS Brascamp-Lieb
inequality, and gives a sharper unentangled bound by setting $b=1$.
\Cref{sec:entanglement} establishes security against entangled leakage
on a specified support and presents the GHZ attack.
We conclude with some open questions in \cref{sec:comparison}.

\section{Preliminaries}
\label{sec:prelim}

\paragraph{Parameters and notation.}
Let $1\leq t\leq n<p$, where $p$ is an odd prime, and fix distinct nonzero
evaluation points $x_1,\ldots,x_n\in\F_p^\times$.  Set
$\omega=e^{2\pi i/p}$.  Inner products $\ip{a}{y}$ over $\F_p^m$
are evaluated in the field.  All logarithms used for rates have base two.
Throughout the paper, all registers that we use are finite-dimensional.

\begin{definition}[Fourier transform]
  \label{def:fourier}
For $g:\F_p^m\to\mathbb C$, its Fourier transform and inversion
formula are given by
\begin{equation*}
  \widehat g(a)=\E_{y\in\F_p^m}g(y)\omega^{-\ip{a}{y}},
  \qquad
  g(y)=\sum_{a\in\F_p^m}\widehat g(a)\omega^{\ip{a}{y}}.
\end{equation*}
\end{definition}

\begin{theorem}[Parseval's identity]
  \label{thm:parseval}
For every $g:\F_p^m\to\mathbb C$,
\begin{equation*}
  \sum_a|\widehat g(a)|^2=\E_y|g(y)|^2.
\end{equation*}
\end{theorem}

\begin{fact}[Conjugate symmetry]
  \label{fact:conjugate}
For real-valued $g$,
\begin{equation*}
  \widehat g(-a)=\overline{\widehat g(a)},
  \qquad |\widehat g(-a)|^2=|\widehat g(a)|^2.
\end{equation*}
\end{fact}

\begin{definition}
    For a linear subspace $V\subseteq\F_p^m$, define its orthogonal complement by $
  V^\perp:=\{a\in\F_p^m:\ip{a}{v}=0\text{ for every }v\in V\}.$
\end{definition}

\begin{fact}[Character average over an affine subspace]
  \label{fact:character}
If $V\subseteq\F_p^m$ is a linear subspace and $b\in\F_p^m$, then
\begin{equation*}
  \E_{y\leftarrow V+b}\omega^{\ip{a}{y}}
  =\begin{cases}
      \omega^{\ip{a}{b}},&a\in V^\perp,\\
      0,&a\notin V^\perp.
    \end{cases}
\end{equation*}
\end{fact}

\begin{fact}[H\"older's inequality]
  \label{fact:holder}
Let $\Omega$ be a finite set, let $R\geq1$ be an integer, and let
$F_1,\ldots,F_R:\Omega\to\mathbb R_{\geq0}$. Then
\[
  \sum_{x\in\Omega}\prod_{r=1}^R F_r(x)^{1/R}
  \leq\prod_{r=1}^R\left(\sum_{x\in\Omega}F_r(x)\right)^{1/R}.
\]
\end{fact}

\begin{fact}[Weighted arithmetic-geometric mean inequality]
  \label{fact:amgm}
For $u,v\geq0$ and $0<\alpha<1$, we have (equality iff $u=v$):
\[
  (1-\alpha)u+\alpha v\geq u^{1-\alpha}v^\alpha.
\]
\end{fact}

\subsection{Codes and Shamir's scheme}

\begin{definition}[MDS codes]
  \label{def:codes}
An $[n,k,d]_p$ linear code is a $k$-dimensional subspace of $\F_p^n$
whose minimum nonzero Hamming weight is $d$.  It is maximum distance separable
(MDS) if $d=n-k+1$.  
\end{definition}

\begin{definition}[Shamir's scheme]
  \label{def:shamir}
For sharing a secret $s\in\F_p$, we sample a uniformly random polynomial $f$ of degree
less than $t$, conditioned on $f(0)=s$, and give party $i$ the share
$y_i=f(x_i)$ for distinct, non-zero $x_i\in \F_p$.  Any $t$ shares reconstruct $f(0)$ by interpolation.
The zero-secret sharing space and the affine sharing coset are
\begin{equation*}
  \C_0=\{(f(x_1),\ldots,f(x_n)):f(0)=0,\ \deg f<t\},
  \qquad \C_s=\C_0+s\one,
\end{equation*}
where $\one=(1,\ldots,1)$.  The sharing vector for secret $s$ is uniform
over $\C_s$.  For $t\geq2$, the code $\C_0$ is an
$[n,t-1,n-t+2]_p$ MDS code; for $t=1$, $\C_0=\{0\}$.
\end{definition}

\begin{definition}[Dual Shamir code]
  \label{def:dual}
The dual Shamir code is
\begin{equation*}
  \C=\C_0^\perp
   =\{a\in\F_p^n:\ip{a}{c}=0\text{ for every }c\in\C_0\}.
\end{equation*}
\end{definition}

\begin{fact}[Dual distance]
  \label{fact:dual-distance}
Let $k=n-t+1$. The dual Shamir code $\C$ is an $[n,k,t]_p$
MDS code. In particular, $|\C|=p^k$, and every nonzero codeword
has Hamming weight at least $t$.
\end{fact}

\begin{fact}[MDS bijectivity]
  \label{fact:mds-bijectivity}
For every $I\subseteq[n]$ with $|I|=n-t+1$, the coordinate
projection $\pi_I:\C\to\F_p^I$ is a bijection.
\end{fact}

\subsection{Distinguishing advantage and leakage models}

\begin{definition}
  \label{def:advantage}
For density operators $\rho,\sigma$, set
\begin{equation*}
  \Adv(\rho,\sigma)=\frac12\|\rho-\sigma\|_1,
  \qquad \|A\|_1=\Tr\sqrt{A^\dagger A}.
\end{equation*}
For probability distributions $P,Q$ on a finite set $\mathcal X$,
their statistical distance is
\begin{equation*}
  \operatorname{SD}(P,Q)
  =\frac12\sum_{x\in\mathcal X}|P(x)-Q(x)|.
\end{equation*}
\end{definition}

\begin{fact}[Data processing]
  \label{fact:data-processing}
For every quantum channel $\Lambda$,
\begin{equation*}
  \Adv\bigl(\Lambda(\rho),\Lambda(\sigma)\bigr)
  \leq\Adv(\rho,\sigma).
\end{equation*}
\end{fact}

\begin{definition}[Leakage device]
    A {leakage device} is a local quantum processor that receives
a party's classical share and produces a quantum output for the
adversary. It may also use auxiliary quantum registers permitted
by the leakage model. Devices may share initial entanglement with the other party's devices or the adversary when the model allows it.
\end{definition}
\begin{definition}[Unentangled local one-qubit leakage]
  \label{def:unentangled-leakage}
For each $i\in[n]$, let $B_i$ be a one-qubit register.  A leakage strategy
specifies a map $y_i\mapsto\rho_i^{B_i}(y_i)$.  These maps are fixed before
the secret and shares are chosen.  Conditioned on $y=(y_1,\ldots,y_n)$, the
leaked state is $\bigotimes_i\rho_i^{B_i}(y_i)$.  The adversary receives
\begin{equation*}
  \rho_s^B=\E_{y\leftarrow\C_s}
       \bigotimes_{i=1}^n\rho_i^{B_i}(y_i)
\end{equation*}
and may perform an arbitrary joint measurement on $B=B_1\cdots B_n$.
\end{definition}

\begin{definition}[Entangled local leakage]
  \label{def:physical-leakage}
Fix $E\subseteq[n]$, $|E|=r$.  Before the shares are drawn, the devices in
$E$ and the adversary prepare a state $\tau^{A_E R}$, where
$A_E=\bigotimes_{i\in E}A_i$ belongs to the devices and $R$ belongs to the
adversary.    The initial state
is independent of the secret and the randomness used to generate the
shares.  Device $i\in E$ applies a local channel
\begin{equation*}
  \mathcal L_{i,y_i}:A_i\longrightarrow B_i,
  \qquad \dim B_i=2.
\end{equation*}
Devices outside $E$ prepare one-qubit states $\rho_i^{B_i}(y_i)$ without a
shared quantum state.  Set
\begin{equation*}
  \sigma_{B_E R}(y_E)=
  \left(\left(\bigotimes_{i\in E}\mathcal L_{i,y_i}\right)
       \otimes\operatorname{id}_R\right)(\tau^{A_E R}),
  \qquad B_E=\bigotimes_{i\in E}B_i.
\end{equation*}
The adversary receives
\begin{equation*}
  \rho_s^{(E,R)}=\E_{y\leftarrow\C_s}
       \sigma_{B_E R}(y_E)\otimes
       \bigotimes_{i\notin E}\rho_i^{B_i}(y_i).
\end{equation*}
The adversary may perform arbitrary joint quantum operations and
measurements on all leaked qubits $B_1\cdots B_n$ together with
its retained register $R$.
The set $E$, the initial state, and all local maps are fixed before the secret
and shares are chosen. 
\end{definition}

\begin{definition}[Disjoint-block quantum leakage]
  \label{def:block-leakage}
Fix an integer $b\geq1$, and let
$\mathcal B=\{B_1,\ldots,B_M\}$ be a partition of $[n]$
into nonempty blocks. Write
\[
  m_j=|B_j|\leq b,\qquad d_j=2^{m_j}.
\]
A leakage strategy specifies a map
$y_{B_j}\mapsto\sigma_j(y_{B_j})$ for each block, where
$\sigma_j(y_{B_j})$ is an arbitrary density operator on
$\mathbb C^{d_j}$, depending jointly on the shares in that
block. Thus each block emits as many qubits as there are
shares in the block. The partition and maps are fixed before
the secret and shares are chosen.

Conditioned on a sharing vector $y$, the leaked state is
$\bigotimes_{j=1}^M\sigma_j(y_{B_j})$. The adversary receives
\begin{equation*}
  \rho_s=\E_{y\leftarrow\C_s}
       \bigotimes_{j=1}^M\sigma_j(y_{B_j})
\end{equation*}
and may measure all outputs jointly. The tensor-product
condition holds for each fixed sharing vector; the average
$\rho_s$ need not be a product state.
\end{definition}

\noindent\emph{Remark:}
If every block contains one share, this model reduces to
\Cref{def:unentangled-leakage}. It also covers local devices that share
entanglement within each block and emit one qubit each, provided different
blocks start independently and the adversary holds no quantum register
correlated with the devices. The combined output of a block is then an
$m_j$-qubit state depending only on that block's shares.

The entangled local model in \Cref{def:physical-leakage} also allows
entanglement between different blocks and with the adversary where each local
device still accesses only its own share. In the block model, a single
leakage map can instead access all shares in its block (in particular, there can be communication between devices within a block). Unentangled local leakage is a special case of both models.

The number of input shares and output qubits can also differ for blocks. Our
analysis also applies with different  bounds on these two quantities,
with adjusted threshold bounds (see the remark at the beginning of
\Cref{sec:bl}).

\begin{definition}[Security]
  \label{def:statistical-security}
Fix a leakage model.  For a legal strategy $\mathcal L$, let
$\rho_s^{\mathcal L}$ be the adversary's state for secret $s$.  The scheme
is $\varepsilon$-secure if, for every legal $\mathcal L$ and every
$s_0,s_1\in\F_p$,
\begin{equation*}
  \Adv\bigl(\rho_{s_0}^{\mathcal L},\rho_{s_1}^{\mathcal L}\bigr)
  \leq\varepsilon.
\end{equation*}  A family is called \textit{exponentially secure} if
its advantage is at most $2^{-cn}$ for some constant $c>0$ and all sufficiently
large $n$.  It has \textit{negligible advantage} if this advantage is at most $n^{-C}$
for every constant $C>0$ and all sufficiently large $n$.  Exponential security
trivially implies negligible advantage. 
\end{definition}

All security bounds also hold with shared classical randomness independent of the secret and shares, even if the adversary knows the seed: condition on the seed, apply the uniform bound, and average using convexity of the trace norm.

\begin{fact}
  \label{fact:pauli-norm}
Let $\Pau_1=\{I,X,Y,Z\}$, and let
$P=(P_1,\ldots,P_n)$ and $Q=(Q_1,\ldots,Q_n)$ be Pauli
strings in $\Pau_1^n$, identified with the operators
$P_1\otimes\cdots\otimes P_n$ and
$Q_1\otimes\cdots\otimes Q_n$, respectively. Then
\[
  \Tr(P^\dagger Q)=2^n\delta_{P,Q},
\]
where $\delta_{P,Q}=1$ if $P=Q$, and
$\delta_{P,Q}=0$ otherwise.
\end{fact}

\section{The Pauli-Fourier analysis}
\label{sec:reduction}
This section develops the Pauli--Fourier reduction used to prove
the elementary $5/6$ threshold in \Cref{sec:cauchy}.
The sharper bounds in \Cref{sec:bl} use a separate, direct analysis
of matrix-valued Fourier coefficients and do not depend on the
results of this section.

Fix an arbitrary unentangled local one-qubit leakage strategy.  Leakage from the
$i$-th share is specified by a function $y\mapsto\rho_i(y)$ from $\F_p$
to one-qubit density operators.  Overall, the adversary receives
\begin{equation}
  \rho_s=\E_{y\leftarrow\C_s}\bigotimes_{i=1}^n\rho_i(y_i).
  \label{eq:product-leakage}
\end{equation}
For security of the scheme, we want to  prove that
$\|\rho_{s_0}-\rho_{s_1}\|_1$ is negligible for any
pair of secrets $(s_0,s_1)$.  We write the leakage corresponding to the share $i$ as
\begin{equation}
  \rho_i(y)=\frac12\bigl(I+r_{i,X}(y)X+r_{i,Y}(y)Y+r_{i,Z}(y)Z\bigr),
  \label{eq:bloch}
\end{equation}
where $r_{i,X}(y),r_{i,Y}(y),r_{i,Z}(y)$ are real and satisfy
$r_{i,X}(y)^2+r_{i,Y}(y)^2+r_{i,Z}(y)^2\leq1$ for every input $y$.
Let $r_{i,I}(y)=1,$
\begin{equation*}
  W_i(a)=\sum_{P\in\Pau_1}|\widehat r_{i,P}(a)|^2.
\end{equation*}

\begin{lemma}[One-qubit Fourier budget]
  \label{lem:qubit-budget}
For every coordinate $i$,
\begin{equation*}
  \sum_a W_i(a)\leq2,
  \qquad 1\leq W_i(0)\leq2,
  \qquad W_i(a)\leq\frac12\quad(a\neq0).
\end{equation*}
\end{lemma}

\begin{proof}

Using \cref{thm:parseval}, we get
\begin{align*}
  \sum_a W_i(a)
  &=\sum_a\bigl(|\widehat r_{i,I}(a)|^2+|\widehat r_{i,X}(a)|^2
       +|\widehat r_{i,Y}(a)|^2+|\widehat r_{i,Z}(a)|^2\bigr)\notag\\
  &=\E_y\bigl[r_{i,I}(y)^2+r_{i,X}(y)^2+r_{i,Y}(y)^2+r_{i,Z}(y)^2\bigr]
    \leq2.
\end{align*}
Since $r_{i,I}\equiv1$, its zero-frequency coefficient is one, so
$1\leq W_i(0)\leq2$.  By \Cref{fact:conjugate}, $W_i(-a)=W_i(a)$.
For $a\neq0$, the two frequencies $a$ and $-a$ are distinct because
$p$ is odd.  Therefore
\begin{equation*}
  2W_i(a)\leq\sum_{b\neq0}W_i(b)\leq2-W_i(0)\leq1.
\end{equation*}
\end{proof}

\begin{theorem}[Hilbert-Schmidt reduction]
  \label{thm:hs}
For every unentangled local one-qubit leakage strategy and every pair of secrets,
\begin{equation*}
  \frac{1}{2}\|\rho_{s_0}-\rho_{s_1}\|_1
  \leq T,
  \qquad
  T:=\sum_{a\in\C\setminus\{0\}}
        \prod_{i=1}^n\sqrt{W_i(a_i)}.
\end{equation*}
\end{theorem}

\begin{proof}
Using \eqref{eq:product-leakage} and \eqref{eq:bloch}, we get
\begin{equation}
  \bigotimes_{i=1}^n\rho_i(y_i)
  =2^{-n}\sum_{P\in\Pau_1^n}
       \left(\prod_{i=1}^n r_{i,P_i}(y_i)\right)P.
  \label{eq:product-pauli}
\end{equation}
 Define $R_{s,P}:=\E_{y\leftarrow\C_s}\prod_{i=1}^n r_{i,P_i}(y_i).$ 
Taking average in \eqref{eq:product-pauli} gives $$
  \rho_s=2^{-n}\sum_{P\in\Pau_1^n}R_{s,P}P.
$$

We next express $R_{s,P}$ in terms of Fourier coefficients.  Fourier inversion applied separately at each coordinate, gives
\begin{align*}
  \prod_{i=1}^n r_{i,P_i}(y_i)
  &=\prod_{i=1}^n\left(
       \sum_{a_i\in\F_p}\widehat r_{i,P_i}(a_i)\omega^{a_i y_i}\right)=\sum_{a\in\F_p^n}
       \left(\prod_{i=1}^n\widehat r_{i,P_i}(a_i)\right)\omega^{\ip{a}{y}}.
\end{align*}
Every $y\in\C_s$ has the unique form $y=c+s\one$, with
$c\in\C_0$, so
\[
  \E_{y\leftarrow\C_s}\omega^{\ip{a}{y}}
  =\omega^{s\ip{a}{\one}}
    \E_{c\leftarrow\C_0}\omega^{\ip{a}{c}}.
\]
Using \Cref{fact:character}, all frequencies outside the dual code vanish, leaving
\begin{equation*}
  R_{s,P}=\sum_{a\in\C}\omega^{s\ip{a}{\one}}
       \prod_{i=1}^n\widehat r_{i,P_i}(a_i).
\end{equation*}
For $a\in\C$, write
\[
  d_a=\omega^{s_0\ip{a}{\one}}-\omega^{s_1\ip{a}{\one}}.
\]
Consider
\[
  \rho_{s_0}-\rho_{s_1}
  =2^{-n}\sum_{a\in\C\setminus\{0\}}d_a
    \sum_{P\in\Pau_1^n}
      \left(\prod_{i=1}^n\widehat r_{i,P_i}(a_i)\right)P,
\]
where the zero frequency cancels because $d_0=0$.

For any operator $A$ on $n$ qubits, with
$\|A\|_2=(\Tr A^\dagger A)^{1/2}$, we have
$\|A\|_1\leq2^{n/2}\|A\|_2$.
We have
\begin{align*}
  \frac12\|\rho_{s_0}-\rho_{s_1}\|_1
  &\leq 2^{-n-1}
    \sum_{a\in\C\setminus\{0\}}|d_a|
    \left\|
      \sum_{P\in\Pau_1^n}
        \left(\prod_{i=1}^n\widehat r_{i,P_i}(a_i)\right)P
    \right\|_1
    &&\mbox{(triangle inequality)}\\
  &\leq 2^{-n/2}
    \sum_{a\in\C\setminus\{0\}}
    \left\|
      \sum_{P\in\Pau_1^n}
        \left(\prod_{i=1}^n\widehat r_{i,P_i}(a_i)\right)P
    \right\|_2
    &&\mbox{($\|A\|_1\leq2^{n/2}\|A\|_2$,~$|d_a|\leq2$)}\\
  &=2^{-n/2}
    \sum_{a\in\C\setminus\{0\}}
    \left(
      2^n\sum_{P\in\Pau_1^n}
        \prod_{i=1}^n|\widehat r_{i,P_i}(a_i)|^2
    \right)^{1/2}
    &&\mbox{(\Cref{fact:pauli-norm})}\\
  &=\sum_{a\in\C\setminus\{0\}}
    \left(
      \sum_{P\in\Pau_1^n}
        \prod_{i=1}^n|\widehat r_{i,P_i}(a_i)|^2
    \right)^{1/2}\\
  &=\sum_{a\in\C\setminus\{0\}}
      \prod_{i=1}^n
      \left(
        \sum_{P_i\in\Pau_1}
          |\widehat r_{i,P_i}(a_i)|^2
      \right)^{1/2}\\
  &=\sum_{a\in\C\setminus\{0\}}
      \prod_{i=1}^n\sqrt{W_i(a_i)}
   =T.
\end{align*}
\end{proof}
\section{An elementary Cauchy-Schwarz bound}
\label{sec:cauchy}
The following proof gives a weaker threshold than the bound in \cref{sec:bl}, but
it is more intuitive and simpler to analyze and uses only Cauchy-Schwarz. This proof is similar to the classical proof in \cite{BDIR2021}, adapted to our quantum leakage setting.

\begin{proposition}[Elementary threshold]
  \label{prop:block}
Suppose $t\geq(2n+2)/3$. For every unentangled local one-qubit leakage
strategy and every pair of secrets $s_0,s_1$,
\begin{equation*}
  \frac12\|\rho_{s_0}-\rho_{s_1}\|_1
  \leq2^{(5n-6t+4)/2}.
\end{equation*}
In particular, for every fixed $\eta>0$, a family with
$t>(5/6+\eta)n$ is exponentially secure against unentangled
local one-qubit leakage.
\end{proposition}

\begin{proof}
Choose pairwise disjoint sets $A_1,A_2,A_3$ whose union is $[n]$, with
\[
  |A_1|=|A_2|=n-t+1,
  \qquad
  |A_3|=n-2(n-t+1)=2t-n-2.
\]
Such sets exist because $t\geq(2n+2)/3$ implies
$
  2t-n-2\geq n-t\geq0.$

By \Cref{fact:dual-distance}, every nonzero codeword has at most
$n-t$ zero coordinates, so at most $n-t$ coordinates in $A_3$
can be zero. Thus \Cref{lem:qubit-budget} gives, for every
$a\in\C\setminus\{0\}$,
\begin{equation}
  \label{eq:elementary-a3-bound}
  \prod_{i\in A_3}\sqrt{W_i(a_i)}
  \leq(\sqrt{2})^{n-t}
       \left(\frac1{\sqrt{2}}\right)^{|A_3|-(n-t)}
  =2^{(3n-4t+2)/2}.
\end{equation}
For each $j\in\{1,2\}$, projection onto $A_j$ is bijective
by \Cref{fact:mds-bijectivity}. Hence \Cref{lem:qubit-budget} gives
\begin{equation}
  \label{eq:elementary-projection-bound}
  \sum_{a\in\C}\prod_{i\in A_j}W_i(a_i)
  =\prod_{i\in A_j}\sum_{b\in\F_p}W_i(b)
  \leq2^{n-t+1}.
\end{equation}
Using \Cref{thm:hs}, we get
\begin{align*}
 &\frac12\|\rho_{s_0}-\rho_{s_1}\|_1
  \\&\leq\sum_{a\in\C\setminus\{0\}}
      \left(\prod_{i\in A_1}\sqrt{W_i(a_i)}\right)
      \left(\prod_{i\in A_2}\sqrt{W_i(a_i)}\right)
      \left(\prod_{i\in A_3}\sqrt{W_i(a_i)}\right)\\
  &\leq2^{(3n-4t+2)/2}
    \sum_{a\in\C\setminus\{0\}}
      \left(\prod_{i\in A_1}\sqrt{W_i(a_i)}\right)
      \left(\prod_{i\in A_2}\sqrt{W_i(a_i)}\right)
    &&\mbox{(\cref{eq:elementary-a3-bound})}\\
  &\leq2^{(3n-4t+2)/2}
    \Big(\sum_{a\in\C\setminus\{0\}}
      \prod_{i\in A_1}W_i(a_i)\Big)^{1/2}
    \Big(\sum_{a\in\C\setminus\{0\}}
      \prod_{i\in A_2}W_i(a_i)\Big)^{1/2}
    &&\mbox{(Cauchy-Schwarz)}\\
  &\leq2^{(3n-4t+2)/2}
    \left(\sum_{a\in\C}
      \prod_{i\in A_1}W_i(a_i)\right)^{1/2}
    \left(\sum_{a\in\C}
      \prod_{i\in A_2}W_i(a_i)\right)^{1/2}
    &&\mbox{($W_i(0)\geq 0$)}\\
  &\leq2^{(3n-4t+2)/2}\,2^{n-t+1}
    &&\mbox{(\cref{eq:elementary-projection-bound})}\\
  &=2^{(5n-6t+4)/2}.
\end{align*}
Thus, $t/n\geq5/6+\eta$ gives exponential security for every fixed $\eta>0$.
\end{proof}

\section{The MDS Brascamp-Lieb bound for disjoint blocks}
\label{sec:bl}
We now prove security in the disjoint-block leakage model of
\Cref{def:block-leakage}. The unentangled case follows by taking
every block to have size one. We use a block version of the MDS
Brascamp-Lieb inequality of Horinaga and
Yamakawa~\cite{HorinagaYamakawa2026}.

Fix a partition $\mathcal B=\{B_1,\ldots,B_M\}$ and a leakage
strategy as in \Cref{def:block-leakage}. Throughout this section,
write
\[
  m_j=|B_j|\leq b,\qquad d_j=2^{m_j},\qquad k=n-t+1.
\]

\paragraph{Remark:}
The same argument applies when block $B_j$ contains $m_j\leq b$
shares and outputs $\ell_j\leq\ell$ qubits, for fixed positive integers
$b,\ell$ independent of $n$, under the same independence assumptions.
Indeed, \Cref{lem:block-fourier-budget} treats the input size and output
dimension separately. Embedding each output in dimension $2^\ell$
replaces $2^b$ by $2^\ell$ in \cref{eq:common-block-bound}, while
$q_j=n/(k+m_j-1)$, $D=\lceil t/b\rceil$, and the MDS argument remain
unchanged. The threshold argument then uses $A_\ell(Q)$ in place of
$A_b(Q)$. Since $0<\kappa_\ell<1$, it gives exponential security
above a suitable threshold below one, with a fixed positive margin and
constants depending on $b,\ell$.

We develop the case $\ell_j=m_j$ explicitly because it gives a total
output budget of one qubit per share, includes local leakage with
entanglement within each block, and recovers the unentangled one-qubit
model for singleton blocks.

\subsection{Fourier bounds for one block}

We first study the output of a single block as a matrix-valued function
of its shares. Its Fourier coefficients describe how this output depends
on linear combinations of those shares. We will bound both the sum of
their squared trace norms and the trace norm of each nonzero coefficient.
Together, these bounds give the higher-power sums needed in the block
MDS inequality.

Let $\mathcal D(\mathbb C^d)$ denote the set of density operators
on $\mathbb C^d$. For an operator-valued map $\sigma:\F_p^m\to\mathcal D(\mathbb C^d)$, define
\begin{equation}
  \widehat\sigma(a)
  =\E_{z\in\F_p^m}\omega^{-\ip{a}{z}}\sigma(z),
  \qquad
  h(a)=\|\widehat\sigma(a)\|_1.
  \label{eq:block-fourier-definition}
\end{equation}
The input has $m$ field coordinates, while the output has dimension
$d$. We keep these parameters separate so that the bounds also apply
when a smaller block state is embedded in a larger output space.
\begin{lemma}[Block Fourier budget]
  \label{lem:block-fourier-budget}
For every map
$\sigma:\F_p^m\to\mathcal D(\mathbb C^d)$, with
$\widehat{\sigma}$ and $h$ defined as in
\Cref{eq:block-fourier-definition}, we have
\[
  h(0)=1,
  \qquad
  \sum_a h(a)^2\leq d,
  \qquad
  \sum_{a\neq0}h(a)^2\leq d-1.
\]
Moreover, for every $a\neq0$, $
  h(a)\leq c_{d,p},$ where
$
  c_{d,p}
  =\min\left\{1,
       \frac d\pi\sin\left(\frac\pi d\right)+\frac{\pi}{p}
       \right\}.$
Hence, for every $q\geq2$,
\[
  \sum_{a\neq0}h(a)^q
  \leq(d-1)c_{d,p}^{\,q-2}.
\] 
\end{lemma}
The proof is deferred to \Cref{app:block-fourier-budget}.

\begin{lemma}
  \label{lem:direct-block-fourier}
For every disjoint-block strategy and any two secrets,
\[
  \Adv(\rho_{s_0},\rho_{s_1})
  \leq\sum_{a\in\C\setminus\{0\}}
       \prod_{j=1}^M\|\widehat\sigma_j(a_{B_j})\|_1.
\]
\end{lemma}

\begin{proof}
The ideas are similar to those in
the proof of \Cref{thm:hs}. Fourier inversion applies entry by entry
to each block output. Thus
\[
  \bigotimes_{j=1}^M\sigma_j(y_{B_j})
  =\sum_{a\in\F_p^n}\omega^{\ip{a}{y}}
       \bigotimes_{j=1}^M\widehat\sigma_j(a_{B_j}).
\]
Averaging over $\C_s=\C_0+s\one$ and applying
\Cref{fact:character} leaves precisely the frequencies in
$\C=\C_0^\perp$:
\[
  \rho_s=\sum_{a\in\C}\omega^{s\ip{a}{\one}}
       \bigotimes_{j=1}^M\widehat\sigma_j(a_{B_j}).
\]
The term at $a=0$ is independent of $s$, so it cancels when the
two states are subtracted. For every other term,
$\left|\omega^{s_0\ip{a}{\one}}-
\omega^{s_1\ip{a}{\one}}\right|\leq2$. Therefore,
\begin{align*}
  \Adv(\rho_{s_0},\rho_{s_1})
  &\leq\frac12\sum_{a\in\C\setminus\{0\}}
    \left|\omega^{s_0\ip{a}{\one}}-
    \omega^{s_1\ip{a}{\one}}\right|
    \left\|\bigotimes_j\widehat\sigma_j(a_{B_j})\right\|_1\\
  &\leq\sum_{a\in\C\setminus\{0\}}
       \prod_j\|\widehat\sigma_j(a_{B_j})\|_1.
\end{align*}
The last line uses multiplicativity of the trace norm under tensor products.
\end{proof}

\subsection{A block MDS inequality}

The following is a variant of the MDS Brascamp-Lieb inequality of
Horinaga and Yamakawa~\cite{HorinagaYamakawa2026} for blocks of different sizes.
Its proof is given in \Cref{app:block-mds}.

\begin{lemma}[MDS Brascamp-Lieb for unequal blocks]
  \label{lem:unequal-block-bl}
Let $k=n-t+1$, and suppose $k+m_j-1\leq n$ for every block.  Define
$
  q_j=\frac{n}{k+m_j-1}.$
For arbitrary nonnegative functions
$g_j:\F_p^{B_j}\to\mathbb R_{\geq0}$,
\[
  \sum_{a\in\C}\prod_{j=1}^M g_j(a_{B_j})
  \leq
  \prod_{j=1}^M
  \Big(\sum_{u\in\F_p^{B_j}}g_j(u)^{q_j}\Big)^{1/q_j}.
\]
\end{lemma}

\subsection{Exponential security for disjoint blocks}

We now combine the Fourier bounds with the block MDS inequality.
For every integer $m\geq1$, define
\begin{equation}\label{eq:defK}
  \kappa_m=\frac{2^m}{\pi}\sin\left(\frac{\pi}{2^m}\right).
\end{equation}
Writing $z=\pi/2^m$, we have $0<z\leq\pi/2$ and
$0<\sin z<z$. Hence $\kappa_m=\sin z/z\in(0,1)$.
This is the limiting coefficient bound for an $m$-qubit output,
since $c_{2^m,p}=\min\{1,\kappa_m+\pi/p\}$.

\begin{theorem}[Exponential security for blocks of bounded size]
  \label{thm:constant-blocks}
For every fixed integer $b\geq1$, there exists a threshold
$\widehat\tau_b\in(1/2,1)$ such that for every fixed $\eta>0$, there are constants $c>0$ and $n_0$
such that every disjoint-block leakage strategy with blocks
of size at most $b$ satisfies
\[
  \Adv(\rho_{s_0},\rho_{s_1})\leq2^{-cn}
\]
whenever $n\geq n_0$ and
$t/n\geq\widehat\tau_b+\eta$.

The threshold $\widehat\tau_b$ is the unique solution in
$(1/2,1)$ of
\begin{equation}
  A_b\left(\frac1{1-\widehat\tau_b}\right)
  =K\left(\frac{\widehat\tau_b}{b}\right),
  \label{eq:unequal-threshold-root}
\end{equation}
where for integers $m\geq1$, real $Q\geq2$, and
$0<\alpha<1$, we define
\begin{equation}\label{eq:defAK}
  A_m(Q)=(2^m-1)\kappa_m^{\,Q-2},
  \qquad
  K(\alpha)=\alpha(1-\alpha)^{(1-\alpha)/\alpha}.
\end{equation}
\end{theorem}

\begin{proof}
By \Cref{lem:direct-block-fourier}, it suffices to bound
\[
  \sum_{a\in\C\setminus\{0\}}\prod_{j=1}^M h_j(a_{B_j}),
  \qquad h_j(u)=\|\widehat\sigma_j(u)\|_1.
\]
Write
\[
  k=n-t+1,\qquad q_j=\frac{n}{k+m_j-1},
  \qquad D=\left\lceil\frac tb\right\rceil.
\]
To bound this sum, we introduce a tilt parameter \(\lambda\geq1\). Choose
$Q\geq2$ such that $q_j\geq Q$ for every block.
Since $k+m_j-1=n/q_j\leq n$, the hypothesis of
\Cref{lem:unequal-block-bl} is satisfied.
Recall that $h_j(0)=1$ for all $j$. We define
\[
  g_j(0)=1,\qquad g_j(u)=\lambda h_j(u)\quad(u\neq0).
\]
Every nonzero codeword of $\C$ has at least $t$ nonzero
coordinates by \Cref{fact:dual-distance}. Since a block contains
at most $b$ coordinates, a nonzero codeword meets at least
$D$ blocks. If it meets $w\geq D$ blocks, then
\begin{equation}
\label{eq:lambdaD}
  \prod_j h_j(a_{B_j})
  =\lambda^{-w}\prod_j g_j(a_{B_j})
  \leq\lambda^{-D}\prod_j g_j(a_{B_j}).
\end{equation}
The above bound holds for every $\lambda\geq1$. We now have
\begin{align}
  \Adv(\rho_{s_0},\rho_{s_1})
  &\leq \sum_{a\in\C\setminus\{0\}}
      \prod_j h_j(a_{B_j})
      &&\mbox{(\Cref{lem:direct-block-fourier})}\nonumber\\ 
  &\leq \lambda^{-D}
      \sum_{a\in\C\setminus\{0\}}
      \prod_j g_j(a_{B_j})
      &&\mbox{(\cref{eq:lambdaD})} \nonumber\\
  &\leq \lambda^{-D}
      \sum_{a\in\C}\prod_j g_j(a_{B_j})
      &&\mbox{(since $g_j(\cdot)\geq0$)} \nonumber\\
  &\leq \lambda^{-D}\prod_{j=1}^M
      \left(\sum_u g_j(u)^{q_j}\right)^{1/q_j}
      &&\mbox{(\Cref{lem:unequal-block-bl}).}\label{eq:tilted-mds-sums}
\end{align}
We will later select an appropriate \(\lambda\geq 1\) to make the upper bound exponentially small. For each block, let
\[
  S_j=\sum_u g_j(u)^Q.
\]
From the monotonicity of $\ell_p$ norm in $p$, we have $\|\mathbf{x}\|_{q_j}\le \|\mathbf{x}\|_{Q}$ when $q_j\ge Q$. For \(\mathbf{x}=(g_j(u))_u\), this gives
\[
  \left(\sum_u g_j(u)^{q_j}\right)^{1/q_j}
  \leq
  \left(\sum_u g_j(u)^Q\right)^{1/Q}
  =S_j^{1/Q}.
\]
Substituting into \cref{eq:tilted-mds-sums}, we obtain
\[
  \Adv(\rho_{s_0},\rho_{s_1})
  \leq\lambda^{-D}\prod_{j=1}^M S_j^{1/Q}.
\]
Since $d_j\leq2^b$, view each block state as a
$2^b\times2^b$ matrix by adding zero rows and columns.
It remains a density matrix, and its Fourier coefficients retain
the same trace norms. Applying \Cref{lem:block-fourier-budget}
in this common dimension gives 
\[
  S_j=1+\lambda^Q\sum_{u\neq0}h_j(u)^Q
  \leq1+(2^b-1)c_{2^b,p}^{\,Q-2}\lambda^Q,
\]
where $
  c_{d,p}
  =\min\left\{1,
       \frac d\pi\sin\left(\frac\pi d\right)+\frac{\pi}{p}
       \right\}$ from \Cref{lem:block-fourier-budget}. Therefore,
\begin{equation}
  \Adv(\rho_{s_0},\rho_{s_1})
  \leq\lambda^{-D}
  \left(1+(2^b-1)c_{2^b,p}^{\,Q-2}\lambda^Q\right)^{M/Q}.
  \label{eq:common-block-bound}
\end{equation} 
To rewrite \cref{eq:common-block-bound}, set
\begin{equation}\label{eq:defA}
  A_{b,p}(Q)=(2^b-1)c_{2^b,p}^{\,Q-2},
  \qquad x=x(\lambda)=\lambda^Q\geq1.
\end{equation}
Let $R_n$ denote the right-hand side of that bound. Since
$\lambda=x^{1/Q}$,
\[
  R_n=x^{-D/Q}\left(1+A_{b,p}(Q)x\right)^{M/Q},
\]
and hence
\[
  \frac{\log_2R_n}{n}
  =\frac1Q\left[-\frac Dn\log_2x
    +\frac Mn\log_2\left(1+A_{b,p}(Q)x\right)\right].
\]
For fixed $x\geq1$ and a fixed $0<\alpha<1$ with
$D/n\geq\alpha$, the inequalities $M\leq n$ and
$\log_2x\geq0$ give the exact bound 
\[
\begin{aligned}
  \frac{\log_2R_n}{n}
  &\leq\frac1Q\left[
       \log_2(1+A_{b,p}(Q)x)-\alpha\log_2x\right]\\
  &=\frac1Q\log_2\left(
       \frac{1+A_{b,p}(Q)x}{x^\alpha}\right).
\end{aligned}
\]
We want the above quantity to be bounded above by a fixed negative
number for an appropriate choice of $x\geq 1$, which gives exponentially small
distinguishing advantage. Recall the definitions from \cref{eq:defK,eq:defAK,eq:defA}. 
For fixed $b,Q$, $A_{b,p}(Q)\to A_b(Q)$ as $n\to\infty$,
uniformly over $p>n$, since $0\le c_{2^b,p}-\kappa_b\le\pi/n$. 
We first examine the above expression with $A=A_b(Q)$, and then
take $n$ sufficiently large.

Fix $A>0$, $Q\geq2$, and $0<\alpha<1$. With $x=\lambda^Q$,
\begin{align}
  -\alpha\log_2\lambda+\frac1Q\log_2(1+A\lambda^Q)
  &=\frac1Q\left[\log_2(1+Ax)-\alpha\log_2x\right] \nonumber\\
  &=\frac1Q\log_2\left(\frac{1+Ax}{x^\alpha}\right). \label{eq:aboveequality}
\end{align}

For every $x>0$, \Cref{fact:amgm} gives
\[
  1+Ax
  =(1-\alpha)\frac1{1-\alpha}+\alpha\frac{Ax}{\alpha}
  \geq\frac{A^\alpha x^\alpha}
              {\alpha^\alpha(1-\alpha)^{1-\alpha}}.
\]
Since \(K(\alpha)=\alpha(1-\alpha)^{(1-\alpha)/\alpha}\),
the denominator on the right is \(K(\alpha)^\alpha\). Thus
\begin{equation}
  \frac{1+Ax}{x^\alpha}
  \geq\left(\frac{A}{K(\alpha)}\right)^\alpha,
  \label{eq:common-amgm}
\end{equation}
with equality at
$ x_*:=\frac{\alpha}{(1-\alpha)A}.
$ 

We now show that \(A<K(\alpha)\) allows an admissible choice
\(x=\lambda^Q\geq1\) for which the distinguishing advantage is exponentially small.
Since $0<1-\alpha<1$ and $(1-\alpha)/\alpha>0$,
\[
  0<(1-\alpha)^{(1-\alpha)/\alpha}<1,
  \qquad 0<K(\alpha)<\alpha<1.
\]
Therefore, if $A<K(\alpha)$, we have
\[
  x_*=\frac{\alpha}{(1-\alpha)A}
  >\frac{\alpha}{(1-\alpha)\alpha}
  =\frac1{1-\alpha}>1.
\]
Since $x=\lambda^Q\geq1$ is equivalent to $\lambda\geq1$, the
scalar exponent is negative for some admissible $\lambda$ when
$A<K(\alpha)$. Under this condition, choose
\begin{equation}
  \lambda=x_*^{1/Q}
  =\left(\frac{\alpha}{(1-\alpha)A}\right)^{1/Q}>1.
  \label{eq:common-tilt-choice}
\end{equation}
The equality in \cref{eq:common-amgm} along with \cref{eq:aboveequality} now gives 
\begin{equation}
\begin{aligned}
  -\alpha\log_2\lambda+\frac1Q\log_2(1+A\lambda^Q)
  &=\frac1Q\log_2\left(\frac{1+Ax_*}{x_*^\alpha}\right)=\frac{\alpha}{Q}\log_2\frac{A}{K(\alpha)}<0.
\end{aligned}
  \label{eq:common-negative-exponent}
\end{equation}
To show that $K$ is strictly increasing, we have:
\[
  \ln K(\alpha)
  =\ln\alpha+\frac{1-\alpha}{\alpha}\ln(1-\alpha).
\]
Differentiation gives
\[
  \frac{d}{d\alpha}\ln K(\alpha)
  =-\frac{\ln(1-\alpha)}{\alpha^2}>0
  \qquad(0<\alpha<1).
\]
Since $K(\alpha)>0$, it follows that
$K'(\alpha)=-K(\alpha)\ln(1-\alpha)/\alpha^2>0$.

As $\tau$ increases from $1/2$ to $1$, the value
$Q=1/(1-\tau)$ increases from $2$ to infinity.
Since $0<\kappa_b<1$, its power $\kappa_b^{Q-2}$
strictly decreases from one to zero. Therefore,
$A_b(1/(1-\tau))=(2^b-1)\kappa_b^{Q-2}$ strictly decreases
from $2^b-1$ to zero. 

Meanwhile, as $\tau$ increases, $\tau/b$ increases, so
$K(\tau/b)$ strictly increases. 
At $\tau=1/2$, we have
$K(1/(2b))<1\leq2^b-1$. Moreover, for every $\tau>1/2$,
\[
  K(\tau/b)\geq K(1/(2b))>0.
\]
Since $A_b(1/(1-\tau))$ tends to zero as
$\tau\to1$, it is eventually smaller than
$K(\tau/b)$.
Since both functions are continuous and monotone, they cross exactly once.
This proves existence and uniqueness in
\cref{eq:unequal-threshold-root}. In particular,
\[
  \tau>\widehat\tau_b
  \quad\Longleftrightarrow\quad
  A_b\left(\frac1{1-\tau}\right)<K\left(\frac\tau b\right).
\]
Fix $\eta>0$. If $\widehat\tau_b+\eta>1$,
the claim is vacuous. Otherwise, set
\[
  \tau_0=\widehat\tau_b+\frac\eta2<1,
  \qquad Q=\frac1{1-\tau_0}>2,
  \qquad \alpha=\frac{\tau_0}{b},
  \qquad A=A_b(Q).
\]
Since $\tau_0>\widehat\tau_b$, we have $A<K(\alpha)$.
From \cref{eq:common-tilt-choice}, let
$x_*=\lambda^Q=\alpha/((1-\alpha)A)$.
For $t/n\geq\widehat\tau_b+\eta$ and $n\geq2b/\eta$, we have
\[
  \frac1{q_j}
  =1-\frac tn+\frac{m_j}{n}
  \leq1-\frac tn+\frac bn
  \leq1-\tau_0=\frac1Q.
\]
Since $q_j\geq Q$, $M\leq n$, and $D/n\geq\alpha$, we can
apply \cref{eq:common-block-bound}. By \cref{eq:common-negative-exponent} and continuity, there is
a constant $c>0$ such that, for all sufficiently large $n$
and every $p>n$,
\begin{equation}
  \frac{\log_2R_n}{n}
  \leq\frac1Q\log_2\left(
       \frac{1+A_{b,p}(Q)x_*}{x_*^\alpha}\right)
  \leq-c.
  \label{eq:finite-negative-exponent}
\end{equation}
Hence
$\Adv(\rho_{s_0},\rho_{s_1})\leq R_n\leq2^{-cn}$, as claimed.

\end{proof}

\subsection{Threshold rates and the unentangled case}

When every block has size $b$, we can apply
\Cref{lem:unequal-block-bl} with
\[
  M=\frac nb,\qquad
  q_j=q:=\frac{n}{k+b-1}
  \quad\text{for every }j.
\]
Suppose $t/n\to\tau\in(1/2,1)$, and put
$Q=1/(1-\tau)$.
Then $q\to Q>2$, so $q\geq2$ for sufficiently
large $n$.

Set
\[
  D=\left\lceil\frac tb\right\rceil,
  \qquad A_{b,p}(q)=(2^b-1)c_{2^b,p}^{\,q-2}.
\]
As before, let $R_n$ denote the right-hand side of
\cref{eq:common-block-bound}, with the  above parameters.
\[
  R_n=\lambda^{-D}
      \left(1+A_{b,p}(q)\lambda^q\right)^{n/(bq)}.
\]
Keep $\lambda\geq1$ fixed, and set $x=\lambda^Q$.
Since $q\to Q$, $D/n\to\tau/b$,
$A_{b,p}(q)\to A_b(Q)$, and $\lambda^q\to x$,
we obtain
\begin{align}
  \lim_{n\to\infty}\frac{\log_2R_n}{n}
  &=\frac1{bQ}
       \left[\log_2(1+A_b(Q)x)-\tau\log_2x\right]=\frac1{bQ}\log_2\left(
       \frac{1+A_b(Q)x}{x^\tau}\right). \label{eq:equal-exponent}
\end{align}
When blocks are of equal size, we know the exact number of blocks which allows us to improve the bound as compared to \Cref{thm:constant-blocks}. This is because \Cref{thm:constant-blocks} uses $M\leq n$ whereas when every block is of size $b$, $M=\frac{n}{b}$. 
\begin{theorem}[Threshold for equal blocks]
  \label{thm:block-threshold}
For each fixed integer $b\geq1$, there is a unique
$Q_b>2$ satisfying
\begin{equation}
  A_b(Q_b)=(Q_b-1)Q_b^{-Q_b/(Q_b-1)}.
  \label{eq:block-threshold-root}
\end{equation}
Set $\tau_b=1-1/Q_b$.
For $\tau\in(1/2,1)$, the expression on the
right-hand side of \cref{eq:equal-exponent} is
negative for some $x\geq1$ if
$\tau>\tau_b$.

For every fixed $\eta>0$, there are constants
$c>0$ and $n_0$ such that equal blocks of size
$b$ satisfy $
  \Adv(\rho_{s_0},\rho_{s_1})\leq2^{-cn}$
whenever $b\mid n$, $n\geq n_0$, and
$t/n\geq\tau_b+\eta$.
\end{theorem}

\begin{proof}
We can apply the parameter calculation from the proof of
\Cref{thm:constant-blocks} to
\cref{eq:equal-exponent}, with
\[
  Q=\frac1{1-\tau},\qquad
  A=A_b(Q),\qquad
  \alpha=\tau.
\]
The factor $1/b$ does not affect the sign.
Thus a choice of $x\geq1$ gives a negative
exponent when $A_b(Q)<K(\tau)$.
Since $\tau=(Q-1)/Q$,
\[
  K(\tau)
  =\frac{Q-1}{Q}\left(\frac1Q\right)^{1/(Q-1)}
  =(Q-1)Q^{-Q/(Q-1)}.
\]

As shown in the previous proof, $K$ strictly
increases, while $A_b(Q)$ strictly decreases with increasing $\tau$.
At $Q=2$, their values are $1/4$ and
$2^b-1$, respectively.
As $Q\to\infty$, we have $A_b(Q)\to0$, whereas
\[
  K(1-1/Q)\geq K(1/2)=1/4.
\]
Continuity and strict monotonicity therefore show that there is a unique crossing.
As before, we get exponentially small distinguishing advantage when the exponent is negative, which happens when $\tau>\tau_b$. This follows by the same continuity and finite-bound argument as in the proof
of \Cref{thm:constant-blocks}, applied to the
equal-block bound.

\end{proof}
The values in \Cref{tab:block-thresholds} are obtained
by solving \cref{eq:block-threshold-root} numerically with the help of ChatGPT Astra.

\begin{table}[htbp]
  \centering
  \renewcommand{\arraystretch}{1.2}
  \setlength{\tabcolsep}{8pt}
  \caption{Approximate values of the sufficient threshold rates for equal independent blocks with one output qubit per share ($\ell=b$).}
  \label{tab:block-thresholds}
  \begin{tabular}{ccc}
    Block size $b$&Output dimension $2^b$&Threshold $\tau_b$\\
    \hline
    1&2&0.7333927776\\
    2&4&0.9331803707\\
    3&8&0.9874928530\\
    4&16&0.9976497793\\
    5&32&0.9995330485\\
    6&64&0.9999031039
  \end{tabular}
\end{table}
\begin{corollary}[Sharper unentangled threshold]
  \label{cor:direct-product}
For every fixed $\eta>0$, the distinguishing advantage for unentangled
local one-qubit leakage is $2^{-\Omega(n)}$ whenever
$t/n\geq\tau_\star+\eta$, where
$\tau_\star:=\tau_1\approx0.7333927776$,
with $\tau_1$ defined in \Cref{thm:block-threshold}.
\end{corollary}
\section{Pre-shared entanglement}
\label{sec:entanglement}

We now consider entangled local leakage as in
\Cref{def:physical-leakage}, where the adversary may jointly process
and measure the leaked qubits together with its retained register $R$.
We bound its distinguishing advantage by reducing to the setting in
which it receives the complete classical shares in the entangled
support $E$, together with the unentangled one-qubit leakage from
the remaining devices. Knowing these shares and the fixed leakage
strategy allows the adversary to prepare the joint state of the
supported devices' outputs and its reference register. Conditioned
on the revealed shares, the remaining problem reduces to unentangled
leakage from a smaller Shamir scheme.

\subsection{Shamir sharing with revealed shares}

\begin{lemma}[Conditioning on revealed shares]
  \label{lem:conditional-shamir}
Let $E\subseteq[n]$ have size $r<t$, and let $Y_E$ denote the shares in $E$.
For every secret $s$, the random variable $Y_E$ is uniform over $\F_p^E$. Fix $z\in\F_p^E$.  Conditioned on $Y_E=z$, the remaining shares are coordinatewise
affine images of a Shamir sharing with $n-r$ parties and threshold $t-r$. More precisely, define $
  L_E(X)=\prod_{i\in E}(X-x_i),$ 
and let $h_z$ be the polynomial of degree less than $r$ satisfying
$h_z(x_i)=z_i$ for every $i\in E$. When $E=\varnothing$, take
$L_E=1$ and $h_z=0$. Then the smaller Shamir secret (for a fixed $s$ and $z$) is
\begin{equation}
  u_{s,z}=\frac{s-h_z(0)}{L_E(0)}.
  \label{eq:conditioned-secret}
\end{equation}
The share from the smaller scheme is transformed into the original share by
\begin{equation*}
  g(x_i)\longmapsto h_z(x_i)+L_E(x_i)g(x_i)=f(x_i),
  \qquad i\notin E,
\end{equation*}
where $g$ is the sharing polynomial of the smaller Shamir scheme.
\end{lemma}

\begin{proof}
Fix $s$ and $z\in\F_p^E$.  A polynomial $f$ agrees with $h_z$ at every
point indexed by $E$ exactly when $f-h_z$ is divisible by $L_E$.  Hence every
polynomial consistent with $Y_E=z$ has a unique representation
\begin{equation}
  f(X)=h_z(X)+L_E(X)g(X).
  \label{eq:conditioned-polynomial}
\end{equation}
Because $\deg L_E=r$ and $\deg h_z<r$, the bound $\deg f<t$ is equivalent
to $\deg g<t-r$. 
The evaluation points are nonzero, so $L_E(0)\neq0$.  Thus $f(0)=s$ is
equivalent to $g(0)=u_{s,z}$, where $u_{s,z}$ is given by
\cref{eq:conditioned-secret}.

There are exactly $p^{t-r-1}$ polynomials $g$ of degree less than $t-r$
with $u_{s,z}$ as the constant term.
From \cref{eq:conditioned-polynomial}, this is also the number of sharing polynomials for the original Shamir scheme consistent with $Y_E=z$.  Every original sharing polynomial
has probability $p^{-(t-1)}$.  Therefore,
\begin{equation}
  \Pr[Y_E=z\mid s]=p^{t-r-1}p^{-(t-1)}=p^{-r},
  \label{eq:revealed-uniform}
\end{equation}
independently of $s$ and $z$.  Conditional on $Y_E=z$, each consistent
polynomial $f$ therefore has probability
$p^{-(t-1)}/p^{-r}=p^{-(t-r-1)}$.  Thus, the non-constant coefficients of $g$ are
independent and uniform, exactly as needed in Shamir sharing. 

Finally, for $i\notin E$, \cref{eq:conditioned-polynomial}
gives $f(x_i)=h_z(x_i)+L_E(x_i)g(x_i)$.  Since the evaluation points are
distinct, $L_E(x_i)\neq0$.  Thus each remaining original share is an invertible affine
image of the corresponding share of $g$.
\end{proof}

\begin{theorem}
  \label{thm:ent-finite}
Let $E\subseteq[n]$ have size $r<t$.  If Shamir sharing with $n-r$ parties
and threshold $t-r$, on the remaining evaluation points, is
$\varepsilon$-secure against unentangled local one-qubit leakage, then the
original scheme is $\varepsilon$-secure against every entangled local
leakage strategy supported on $E$.
\end{theorem}
\begin{proof}
Fix an arbitrary entangled local leakage strategy supported on $E$,
as in \Cref{def:physical-leakage}, and write its leaked state as
$\rho_s^{(E,R)}$. The strategy specifies the initial shared state
$\tau^{A_E R}$ and the local leakage channels
$\mathcal L_{i,y_i}$. Together with the share vector $z$,
these determine the state $\sigma_{B_E R}(z)$. Let
\begin{equation*}
  \overline\rho_{s,z}
  =
  \E_{Y\leftarrow\C_s}\left[
    \bigotimes_{i\notin E}\rho_i(Y_i)
    \,\middle|\,
    Y_E=z
  \right]
\end{equation*}
be the leakage from the devices outside $E$, conditioned on secret $s$ and
$Y_E=z$.  By \cref{eq:revealed-uniform},
\begin{equation*}
  \rho_s^{(E,R)}
  =\frac1{p^r}\sum_{z\in\F_p^E}
    \sigma_{B_E R}(z)\otimes\overline\rho_{s,z}.
\end{equation*}
Consider instead the state in which the adversary receives the complete
classical shares in $E$, together with the leakage from the remaining
devices:
\begin{equation*}
  \widetilde\rho_s
  =\frac1{p^r}\sum_{z\in\F_p^E}
    |z\rangle\!\langle z|^Z\otimes\overline\rho_{s,z}.
\end{equation*}
There is a quantum channel that measures $Z$ in its computational
basis, prepares $\sigma_{B_E R}(z)$ on outcome $z$, and discards
$Z$. It acts as the identity on the outside leakage registers.
Preparing $\sigma_{B_E R}(z)$ requires only $z$ and the fixed
leakage strategy, including its initial shared state
$\tau^{A_E R}$ and its local leakage channels. Thus the same channel applies to all
secrets and maps
$\widetilde\rho_s$ to $\rho_s^{(E,R)}$.
Therefore,
\begin{align}
  \Adv(\rho_{s_0}^{(E,R)},\rho_{s_1}^{(E,R)})
  &\leq\Adv(\widetilde\rho_{s_0},\widetilde\rho_{s_1})
    &&\mbox{(\Cref{fact:data-processing})}\notag\\
  &=\frac1{p^r}\sum_{z\in\F_p^E}
    \Adv(\overline\rho_{s_0,z},\overline\rho_{s_1,z})
    &&\mbox{(block-diagonal states)}.
  \label{eq:revelation-average-bound}
\end{align}
Fix $z$. By \Cref{lem:conditional-shamir}, the conditioned shares
outside $E$ are affine images of a smaller Shamir sharing with
$n-r$ parties, threshold $t-r$, and secret $u_{s,z}$.
For every $i\notin E$, define the leakage function
\begin{equation*}
  \rho_{i,z}'(u)
  =\rho_i\bigl(h_z(x_i)+L_E(x_i)u\bigr),
  \qquad u\in\F_p.
\end{equation*}
If $g$ is the sharing polynomial of the smaller scheme, then
\[
  \rho_{i,z}'(g(x_i))
  =\rho_i\bigl(h_z(x_i)+L_E(x_i)g(x_i)\bigr)
  =\rho_i(f(x_i)).
\]
Thus applying $\rho_{i,z}'$ to the smaller share produces exactly
the leakage obtained by applying $\rho_i$ to the original share.
 The constants
in this computation depend on $z$ and the public evaluation points, but not on
the secret.  Thus the same functions $\rho_{i,z}'$ are used for all secrets and are fixed before the smaller sharing is sampled. So they form a valid unentangled local one-qubit
leakage strategy.

Writing $\C'_u$ for the sharing coset of the smaller scheme with
secret $u$, \Cref{lem:conditional-shamir} gives
\begin{equation*}
  \overline\rho_{s,z}
  =\E_{v\leftarrow\C'_{u_{s,z}}}
      \bigotimes_{i\notin E}\rho_{i,z}'(v_i),
\end{equation*}
where $v_i=g(x_i)$ is the smaller share at coordinate $i$.
By the assumed security of the smaller scheme, applied to the
leakage functions $\rho_{i,z}'$ and the secrets
$u_{s_0,z}$ and $u_{s_1,z}$, we obtain
$\Adv(\overline\rho_{s_0,z},\overline\rho_{s_1,z})
  \leq\varepsilon$
for every $z$. \cref{eq:revelation-average-bound} then gives
$\Adv(\rho_{s_0}^{(E,R)},\rho_{s_1}^{(E,R)})
  \leq\varepsilon.$
\end{proof}

\subsection{Security threshold}

\begin{corollary}
  \label{cor:ent-region}
Write $\tau=t/n$ and $\delta=r/n$. Suppose there are
constants $c_0,\eta>0$, independent of $n$, such that,
for all sufficiently large $n$,
\[
  1-\delta\geq c_0,
  \qquad
  \frac{\tau-\delta}{1-\delta}
    \geq\tau_\star+\eta.
\]
Then Shamir sharing is exponentially secure against entangled local
leakage supported on any set of at most $r$ devices.
In particular, it is exponentially secure when $r=o(n)$ and
$t/n>\tau_\star$ with a fixed positive gap.
\end{corollary}

\begin{proof}
The rate condition implies $r<t$. By \Cref{thm:ent-finite},
it suffices to prove unentangled leakage security for the smaller
scheme with length $N=n-r$, threshold $T=t-r$, and threshold rate
\begin{equation*}
  \frac TN=\frac{t-r}{n-r}
  =\frac{\tau-\delta}{1-\delta}\geq\tau_\star+\eta.
\end{equation*}
By \Cref{thm:block-threshold} with $b=1$, its unentangled leakage
advantage is $2^{-\Omega(N)}$. Since $N=(1-\delta)n\geq c_0n$,
this is $2^{-\Omega(n)}$. Applying \Cref{thm:ent-finite} gives the same
bound for entangled local leakage on any support of size $r$.

Any entangled local leakage strategy supported on fewer than $r$
devices can be represented as one supported on a set of size $r$.
Give each added device an independent ancilla and let its channel
ignore that ancilla and reproduce its original one-qubit leakage.
Hence the bound holds for every support of size at most $r$.

If $r=o(n)$, then for large enough $n$, we have 
\[
  \frac{t-r}{n-r}
  =\frac tn-o(1).
\]
A fixed positive gap between $t/n$ and $\tau_\star$ therefore gives
security for sufficiently large $n$.
\end{proof}

For fixed rates $\tau>\tau_\star$ and $\delta<1$,
the above condition is equivalent to
\begin{equation}
  \delta<\delta_{\max}(\tau)
  :=\frac{\tau-\tau_\star}{1-\tau_\star}.
  \label{eq:support-endpoint}
\end{equation}
\Cref{tab:support-rates} lists the entanglement-support endpoints for exponential security certified
by the full-revelation reduction for several fixed threshold rates (calculated with the help of ChatGPT Astra). Entries are rounded to three decimal places and security holds at
every fixed support rate strictly below the exact endpoint in
\cref{eq:support-endpoint}. These endpoints need not be optimal.
\begin{table}[htbp]
  \centering
  \renewcommand{\arraystretch}{1}
  \setlength{\tabcolsep}{6pt}
  \caption{Entanglement-support endpoints according to full-revelation
  of the shares.}
  \label{tab:support-rates}
  \begin{tabular}{c|ccccc}
    Threshold rate $\tau$ & $0.76$ & $0.80$ & $0.85$ & $0.90$ & $0.95$\\
    \hline
    Entanglement-support endpoint $\delta_{\max}(\tau)$
      & $0.100$ & $0.250$ & $0.437$ & $0.625$ & $0.812$
  \end{tabular}
\end{table}

At $r=t$, revealing the supported shares reconstructs the secret
by \Cref{def:shamir}, so
\cref{thm:ent-finite} cannot give security.  The next attack shows that $t$
devices can obtain almost the same distinguishing power while emitting only one
classical bit each.

\subsection{A classical-leakage \texorpdfstring{$t$}{t}-party GHZ attack}
The attack below applies the GHZ phase-encoding and parity method used by
Buhrman, van Dam, H{\o}yer, and Tapp~\cite[Section~3.1]{BuhrmanVanDamHoyerTapp1999}
to Shamir reconstruction. 

\begin{theorem}[GHZ parity attack]
  \label{thm:ghz}
For every $T\subseteq[n]$ with $|T|=t$, there is a leakage attack using a
$t$-qubit GHZ state shared by the devices in $T$.  Each device leaks one
classical bit locally.  For secrets $s_0=0$ and
$s_1=(p-1)/2$, the leaked states satisfy $
  \Adv(\rho_{s_0},\rho_{s_1})
  \geq\cos^2\left(\frac{\pi}{2p}\right)=1-O(p^{-2}).$
\end{theorem}

\begin{proof}
Fix $T\subseteq[n]$ of size $t$.  By Lagrange interpolation, there are
coefficients $\lambda_i\in\F_p$, for $i\in T$, such that every polynomial
$f$ of degree less than $t$ satisfies
\begin{equation*}
  f(0)=\sum_{i\in T}\lambda_i f(x_i),
  \qquad
  \lambda_i=\prod_{\substack{j\in T\\j\neq i}}
       \frac{-x_j}{x_i-x_j}.
\end{equation*}
These coefficients depend only on the public evaluation points.  Thus, for
shares $y_i=f(x_i)$,
\begin{equation}
  s=\sum_{i\in T}\lambda_i y_i.
  \label{eq:ghz-reconstruction}
\end{equation}

The leakage devices corresponding to the parties in $T$ share the $t$-qubit
GHZ state
$|\mathrm{GHZ}\rangle=\frac{|0^t\rangle+|1^t\rangle}{\sqrt2}.$
Each party in $T$ holds one qubit of this state.  Parties outside $T$ simply
output a fixed bit, say $+1$.

When party $i\in T$ receives its share $y_i$, it applies the local unitary
\[
  U_i(y_i)=|0\rangle\!\langle0|+
           \omega^{\lambda_i y_i}|1\rangle\!\langle1|
\]
to its qubit.  These gates leave $|0^t\rangle$ unchanged.  On
$|1^t\rangle$, their phase factors multiply to
$\prod_{i\in T}\omega^{\lambda_i y_i}
=\omega^{\sum_{i\in T}\lambda_i y_i}=\omega^s$, by
\cref{eq:ghz-reconstruction}.  Hence the shared state becomes
\begin{equation*}
  |\psi_s\rangle=\frac{|0^t\rangle+\omega^s|1^t\rangle}{\sqrt2}.
\end{equation*}
This state is the same for every sharing polynomial with secret $s$, so averaging
over the sharing randomness leaves it unchanged.

Now each party in $T$ measures its qubit in the Pauli-$X$ basis and outputs
the classical measurement result $b_i\in\{\pm1\}$.   The adversary receives
the $t$ classical bits and computes their product $B=\prod_{i\in T}b_i$. This product is exactly the measurement outcome of the observable $X^{\otimes t}$. Therefore
$\E[B\mid s]=\langle\psi_s|X^{\otimes t}|\psi_s\rangle$.

Note that $X^{\otimes t}$ interchanges
$|0^t\rangle$ and $|1^t\rangle$.  Thus,
\begin{align*}
  \E[B\mid s]
  &=\frac12\bigl(\langle0^t|+\omega^{-s}\langle1^t|\bigr)
                 \bigl(|1^t\rangle+\omega^s|0^t\rangle\bigr)\notag\\
  &=\frac{\omega^s+\omega^{-s}}2
   =\operatorname{Re}(\omega^s)
   =\cos\left(\frac{2\pi s}{p}\right).
\end{align*}
Since $B\in\{\pm1\}$, we have:
\[
  \Pr[B=+1\mid s]
  =\frac{1+\cos(2\pi s/p)}{2},
  \qquad
  \Pr[B=-1\mid s]
  =\frac{1-\cos(2\pi s/p)}{2}.
\]
For $s_0=0$, we have $\E[B\mid s_0]=1.$
For $s_1=(p-1)/2$, we have
$\E[B\mid s_1]
=\cos\left(\pi-\frac{\pi}{p}\right)
=-\cos\left(\frac{\pi}{p}\right).$
Since the adversary can compute $B$ from the leaked bits,
\[
  \Adv(\rho_{s_0},\rho_{s_1})
  \geq \operatorname{SD}(B\mid s_0,\,B\mid s_1)
  =\frac{1+\cos(\pi/p)}{2}
  =\cos^2\left(\frac{\pi}{2p}\right)
  =1-O(p^{-2}).
\]
This proves the claimed attack. Note that the same idea distinguishes any fixed pair $s_0\neq s_1$
with the same advantage by choosing $k\in\F_p^\times$ such that
$k(s_1-s_0)=(p-1)/2$, replacing each phase by
$\omega^{k\lambda_i y_i}$, and applying an additional phase
$\omega^{-ks_0}$ to $|1\rangle$ at one participating device.
\end{proof}

\paragraph{Difference from ordinary classical local leakage.}
In the usual classical one-bit local leakage model, the leakage has
the form $
  \bigl(g_1(y_1),\ldots,g_n(y_n)\bigr),$
where each $g_i:\F_p\to\{\pm1\}$ is a local function of the
corresponding share.
In the attack above, the outputs are generated by local measurements
on a previously entangled quantum state.
Although each output is a single classical bit, the attack uses
entangled leakage devices and therefore does not contradict security
results for ordinary classical local leakage.

For every fixed threshold rate $\tau>\tau_\star$,
\Cref{cor:ent-region} certifies security for a positive linear number of
devices sharing arbitrary entanglement with one another and with the
adversary, while the attack above uses $t=\Theta(n)$ devices and needs
no adversary-held reference. Thus the achievable entanglement-support
size in this model is determined up to constant factors in this regime. The optimal support fraction remains
open. In particular, for
\(1\leq\ell<t\), the attack extends to a GHZ state on
\(t-\ell\) devices: the other \(\ell\) devices output locally
randomized, share-dependent bits, giving distinguishing advantage at least
\(2^{-\ell}\cos^2(\pi/(2p))\). Determining the optimal support fraction is left to future work. 
\section{Comparison and further questions}
\label{sec:comparison}

For worst-case classical one-bit leakage, security is known for threshold
rates above approximately $0.668$~\cite{Kasser2024}. Our sufficient
threshold for unentangled one-qubit leakage is
$\tau_\star\approx0.73339$.
This comparison and our results on entangled leakage suggest the
following questions.

\begin{enumerate}
  \item Can one lower the sufficient threshold for unentangled one-qubit
  leakage to match the classical bound of approximately $0.668$,
  or even improve upon it? Since classical one-bit leakage is a special
  case of one-qubit leakage, a better bound below  $0.668$ would also
  improve the known classical bound.

  \item Can one close the gap between the certified entanglement-support
  region in \Cref{cor:ent-region} and the $t$-party GHZ attack?
  More precisely, for a given threshold rate, what is the largest
  fraction of devices that can share arbitrary entanglement with one
  another and with the adversary while preserving security?

 \item Can one use the actual structure on the entangled leakage devices to do a better analysis than revealing their shares completely?

  \item Is the quantum polynomial secret-sharing scheme of
  Cleve-Gottesman-Lo~\cite{CleveGottesmanLo1999}, a quantum analogue
  of Shamir's scheme, resilient to local quantum leakage?
  Here both the secret and the shares are quantum states, so the
  classical conditioning argument in \Cref{thm:ent-finite} does not
  directly apply.
\end{enumerate}

\section*{Disclosure of AI assistance}
ChatGPT (GPT-5.5, GPT-5.6 and Astra) helped in identifying the GHZ leakage attack and assisted in the analysis and numerical optimisation of the Fourier bounds and MDS inequalities. ChatGPT was also used to assist with drafting the manuscript, creating the numerical tables, and basic editing.
The authors verified, extended, and reorganized these arguments to obtain the final results and retain all responsibility for the contents of the paper.

\bibliographystyle{splncs04}
\bibliography{references}

\appendix
\section{Proofs of the matrix facts}
\label{app:matrix-facts}

\begin{fact}[Trace-norm duality]
  \label{fact:trace-duality}
For every matrix $A\in\mathbb C^{d\times d}$,
\[
  \|A\|_1=\max_{U\text{ unitary}}\operatorname{Re}\Tr(UA).
\]
\end{fact}

\begin{proof}
Let $A=V\Sigma W^\dagger$ be a
singular-value decomposition, where the diagonal entries of
$\Sigma$ are $s_1,\ldots,s_d$.
For every unitary $U$,
\[
\begin{aligned}
  \operatorname{Re}\Tr(UA)
  &=
  \operatorname{Re}\Tr
       (W^\dagger UV\Sigma)\\
  &=
  \sum_{\ell=1}^d
       s_\ell\operatorname{Re}(W^\dagger UV)_{\ell\ell}\\
  &\le
  \sum_{\ell=1}^d s_\ell
  =
  \|A\|_1.
\end{aligned}
\]
Here $W^\dagger UV$ is unitary, so each of its diagonal
entries has absolute value at most one.
Equality holds for $U=WV^\dagger$, proving the identity.
\end{proof}

\begin{fact}[Expectation of a Hermitian matrix]
  \label{fact:hermitian-expectation}
For every Hermitian matrix $H$ and density matrix $\tau$ of the
same dimension,
\[
  \Tr(H\tau)\leq\lambda_{\max}(H).
\]
\end{fact}

\begin{proof}
Write
\[
  H=\sum_\ell\eta_\ell
       |v_\ell\rangle\langle v_\ell|.
\]
Then
\[
  \Tr(H\tau)
  =
  \sum_\ell\eta_\ell
       \langle v_\ell|\tau|v_\ell\rangle
\]
is a weighted average of the eigenvalues $\eta_\ell$:
the weights are nonnegative and sum to
$\Tr(\tau)=1$.
It is therefore at most $\max_\ell\eta_\ell=\lambda_{\max}(H)$.
\end{proof}

\section{Proof of the block Fourier budget}
\label{app:block-fourier-budget}

To bound a nonzero coefficient, we will compare an average over equally
spaced points on the circle with an integral. The following Riemann-sum
estimate controls the error in that comparison.

A function $f:[A,B]\to\mathbb R$ is $L$-Lipschitz if
\[
  |f(x)-f(y)|\leq L|x-y|\qquad\text{for all }x,y\in[A,B].
\]

\begin{fact}[Riemann-sum error]
  \label{fact:riemann-error}
Let $A<B$, let $f:[A,B]\to\mathbb R$ be $L$-Lipschitz, and
let $N\geq1$ be an integer. Put $z_j=A+j(B-A)/N$. Then
\[
  \left|\int_A^B f(x)\,dx-\frac{B-A}{N}\sum_{j=0}^{N-1}f(z_j)\right|
  \leq\frac{L(B-A)^2}{2N}.
\]
\end{fact}

\begin{proof}[\Cref{lem:block-fourier-budget}]
Recall that
\[
  \widehat{\sigma}(a)
  =
  \E_{z\in\F_p^m}
  \omega^{-\langle a,z\rangle}\sigma(z),
  \qquad
  h(a)=\|\widehat{\sigma}(a)\|_1,
  \qquad
  \omega=e^{2\pi i/p}.
\]
All expectations over finite sets are uniform.

At $a=0$, the Fourier phase is identically one, so
\[
  \widehat{\sigma}(0)=\E_z\sigma(z).
\]
An average of density matrices is a density matrix. Since the trace
norm of a positive semidefinite matrix equals its trace, we obtain
\[
  h(0)
  =
  \|\widehat{\sigma}(0)\|_1
  =
  \Tr(\widehat{\sigma}(0))
  =
  1.
\]

We next bound the sum of the squared Hilbert-Schmidt norms.
The Fourier transform of a matrix-valued function acts entry by
entry:
\[
  \bigl(\widehat{\sigma}(a)\bigr)_{rs}
  =
  \E_z
  \omega^{-\langle a,z\rangle}\sigma_{rs}(z)
  =
  \widehat{\sigma_{rs}}(a).
\]
For any matrix $A$,
\begin{equation}
  \|A\|_2^2=\Tr(A^\dagger A)=\sum_{r,s}|A_{rs}|^2.
  \label{eq:hs-entrywise}
\end{equation}
It follows that
\begin{align*}
  \sum_a\|\widehat\sigma(a)\|_2^2
  &=\sum_{r,s}\sum_a|\widehat{\sigma_{rs}}(a)|^2
    &&\mbox{(\Cref{eq:hs-entrywise})}\\
  &=\sum_{r,s}\E_z|\sigma_{rs}(z)|^2
    &&\mbox{(\Cref{thm:parseval})}\\
  &=\E_z\Tr(\sigma(z)^\dagger\sigma(z))
    &&\mbox{(\Cref{eq:hs-entrywise})}\\
  &=\E_z\Tr(\sigma(z)^2)
    &&\mbox{(Hermiticity of $\sigma(z)$)}.
\end{align*}

To bound this expectation, let
$\mu_1,\ldots,\mu_d$ be the eigenvalues of $\sigma(z)$.
Since $\sigma(z)$ is a density matrix,
\[
  \mu_\ell\ge0,
  \qquad
  \sum_{\ell=1}^d\mu_\ell=1.
\]
Thus $0\le\mu_\ell\le1$, and hence
\[
  \Tr(\sigma(z)^2)
  =
  \sum_{\ell=1}^d\mu_\ell^2
  \le
  \sum_{\ell=1}^d\mu_\ell
  =
  1.
\]
Therefore,
\[
  \sum_a\|\widehat{\sigma}(a)\|_2^2\le1.
\]

For every $d\times d$ matrix $A$, Cauchy-Schwarz gives
\[
  \|A\|_1^2\leq d\|A\|_2^2.
\]
Applying this inequality to every Fourier coefficient yields
\[
  \sum_a h(a)^2
  \le
  d\sum_a\|\widehat{\sigma}(a)\|_2^2
  \le d.
\]
Since $h(0)=1$, it follows that
\[
  \sum_{a\ne0}h(a)^2
  =
  \sum_a h(a)^2-h(0)^2
  \le d-1.
\]
For any $a\ne0$, the linear function $z\mapsto\langle a,z\rangle$ takes each
value $u\in\F_p$ exactly $p^{m-1}$ times.

For each $u\in\F_p$, define
\[
  \tau_u
  =
  \E\bigl[
    \sigma(z)\mid\langle a,z\rangle=u
  \bigr].
\]
Each $\tau_u$ is a density matrix. Grouping the Fourier sum
according to the value of $\langle a,z\rangle$ gives
\[
  \widehat{\sigma}(a)
  =
  \frac1p\sum_{u=0}^{p-1}\omega^{-u}\tau_u.
\]

Applying \Cref{fact:trace-duality} to $\widehat{\sigma}(a)$, we obtain
\[
  h(a)
  =
  \max_{U\text{ unitary}}
  \frac1p\sum_{u=0}^{p-1}
  \operatorname{Re}\Tr
       (\omega^{-u}U\tau_u).
\]
Fix a unitary $U$, and define the Hermitian part of
$\omega^{-u}U$ as
\[
  H_u
  =
  \frac{\omega^{-u}U+\omega^uU^\dagger}{2}.
\]
Since $\tau_u^\dagger=\tau_u$,
\begin{align*}
  \overline{\Tr(\omega^{-u}U\tau_u)}
  &=\Tr\bigl((\omega^{-u}U\tau_u)^\dagger\bigr)
    &&\mbox{(conjugating the trace)}\\
  &=\Tr(\omega^u\tau_uU^\dagger)
    &&\mbox{($\tau_u^\dagger=\tau_u$)}\\
  &=\Tr(\omega^uU^\dagger\tau_u)
    &&\mbox{(cyclicity)}.
\end{align*}
Therefore,
\[
\begin{aligned}
  \operatorname{Re}\Tr(\omega^{-u}U\tau_u)
  &=
  \frac12\left(
    \Tr(\omega^{-u}U\tau_u)
    +
    \Tr(\omega^uU^\dagger\tau_u)
  \right)\\
  &=
  \Tr(H_u\tau_u).
\end{aligned}
\]
Applying \Cref{fact:hermitian-expectation} to $H_u$ and $\tau_u$ yields
\[
  h(a)
  \le
  \max_{U\text{ unitary}}
  \frac1p\sum_{u=0}^{p-1}\lambda_{\max}(H_u).
\]

Write the spectral decomposition of $U$ as
\[
  U
  =
  \sum_{\ell=1}^d e^{i\theta_\ell}
       |v_\ell\rangle\langle v_\ell|.
\]
Then
\[
  H_u
  =
  \sum_{\ell=1}^d
  \cos\left(\theta_\ell-\frac{2\pi u}{p}\right)
       |v_\ell\rangle\langle v_\ell|.
\]
Thus, defining the $2\pi$-periodic function
\[
  f(\theta)
  =
  \max_{\ell\in[d]}\cos(\theta-\theta_\ell),
\]
we have
\[
  \lambda_{\max}(H_u)
  =
  f\left(\frac{2\pi u}{p}\right).
\]
We will bound the discrete average of $f$ by a quantity
independent of the choice of $U$.

First consider the integral average
\begin{equation}
  I=\frac1{2\pi}\int_0^{2\pi}f(\theta)\,d\theta.
  \label{eq:phase-integral}
\end{equation}
Order the eigenvalue phases so that
\[
  0\le\theta_1\le\cdots\le\theta_d<2\pi,
\]
and set
\[
  \theta_{d+1}=\theta_1+2\pi,
  \qquad
  \Delta_\ell=\theta_{\ell+1}-\theta_\ell.
\]
Repeated phases are allowed and give zero gaps. We have
\[
  \Delta_\ell\ge0,
  \qquad
  \sum_{\ell=1}^d\Delta_\ell=2\pi.
\]

The value of $f(\theta)$ is the cosine of the smallest angular
distance to an eigenvalue phase on the circle. This follows because
cosine decreases with angular distance in $[0,\pi]$.

Consider the gap from $\theta_\ell$ to
$\theta_{\ell+1}$, and write
$\theta=\theta_\ell+x$, where
$0\le x\le\Delta_\ell$.
On the first half of the gap, a nearest phase is
$\theta_\ell$; on the second half, a nearest phase is
$\theta_{\ell+1}$. Therefore,
\[
  f(\theta_\ell+x)
  =
  \begin{cases}
    \cos x,
      &0\le x\le\Delta_\ell/2,\\
    \cos(\Delta_\ell-x),
      &\Delta_\ell/2\le x\le\Delta_\ell.
  \end{cases}
\]
The integral over this gap is therefore
\[
\begin{aligned}
  \int_{\theta_\ell}^{\theta_{\ell+1}}
       f(\theta)\,d\theta
  &=
  \int_0^{\Delta_\ell/2}\cos x\,dx
  +
  \int_{\Delta_\ell/2}^{\Delta_\ell}
       \cos(\Delta_\ell-x)\,dx\\
  &=
  2\int_0^{\Delta_\ell/2}\cos x\,dx\\
  &=
  2\sin(\Delta_\ell/2).
\end{aligned}
\]
Summing over all gaps and using periodicity in \Cref{eq:phase-integral} gives
\[
  I
  =
  \frac1\pi
  \sum_{\ell=1}^d\sin(\Delta_\ell/2).
\]
The function sine is concave on $[0,\pi]$, and each
$\Delta_\ell/2$ lies in this interval. Jensen's inequality
therefore gives
\[
\begin{aligned}
  \frac1d\sum_{\ell=1}^d\sin(\Delta_\ell/2)
  &\le
  \sin\left(
    \frac1d\sum_{\ell=1}^d\frac{\Delta_\ell}{2}
  \right)\\
  &=
  \sin(\pi/d).
\end{aligned}
\]
Hence
\[
  I\le\frac d\pi\sin(\pi/d).
\]

Each shifted cosine is one-Lipschitz, since its derivative has absolute
value at most one. The maximum of finitely many one-Lipschitz functions is
also one-Lipschitz, so $f$ is one-Lipschitz. Set
\[
  \alpha_u=\frac{2\pi u}{p},\qquad
  S=\frac1p\sum_{u=0}^{p-1}f(\alpha_u).
\]
The points $\alpha_u$, for $u=0,\ldots,p-1$, are the left endpoints
of the $p$ equal subintervals of $[0,2\pi]$.
By \Cref{fact:riemann-error} with $L=1$ and $N=p$,
\[
  |S-I|
  =\frac1{2\pi}\left|
    \frac{2\pi}{p}\sum_{u=0}^{p-1}f(\alpha_u)
    -\int_0^{2\pi}f(\theta)\,d\theta
  \right|
  \leq\frac1{2\pi}\frac{(2\pi)^2}{2p}
  =\frac{\pi}{p}.
\]
Combining this with the integral estimate gives
\[
  S
  \le
  \frac d\pi\sin(\pi/d)+\frac{\pi}{p}.
\]
The right-hand side does not depend on $U$. Thus the
earlier trace-norm bound implies
\[
  h(a)
  \le
  \frac d\pi\sin(\pi/d)+\frac{\pi}{p}.
\]

We also have the elementary bound
\[
\begin{aligned}
  h(a)
  &=
  \left\|\frac1p\sum_{u=0}^{p-1}
       \omega^{-u}\tau_u\right\|_1\\
  &\le
  \frac1p\sum_{u=0}^{p-1}
       |\omega^{-u}|\,\|\tau_u\|_1\\
  &=1.
\end{aligned}
\]
Taking the smaller of these two bounds proves
\[
  h(a)\le
  \min\left\{
    1,\,
    \frac d\pi\sin(\pi/d)+\frac{\pi}{p}
  \right\}
  =
  c_{d,p}
  \qquad(a\ne0).
\]

For every $q\geq2$, the pointwise estimate gives
\[
  h(a)^q\leq h(a)^2c_{d,p}^{\,q-2}\qquad(a\ne0).
\]
Summing over nonzero frequencies and using the squared
Fourier-mass bound, we conclude that
\[
  \sum_{a\ne0}h(a)^q
  \le
  c_{d,p}^{\,q-2}\sum_{a\ne0}h(a)^2
  \le
  (d-1)c_{d,p}^{\,q-2}.
\]
\end{proof}

\section{Proof of the block MDS inequality}
\label{app:block-mds}

\begin{proof}[\Cref{lem:unequal-block-bl}]
Let $\C\subseteq\F_p^n$ be an MDS code of dimension
$k$, and let $B_1,\ldots,B_M$ be the given partition
of the coordinates. Write
\[
  m_j=|B_j|,
  \qquad
  q_j=\frac{n}{k+m_j-1}.
\]
By assumption, $k+m_j-1\le n$ for every $j$.

Since $\C$ is MDS, its minimum distance is
$n-k+1$. If two codewords agree on at least $k$
coordinates, they differ in at most $n-k$ coordinates.
The minimum-distance property therefore forces them
to be equal.

Therefore, for every $T\subseteq[n]$ with
$|T|\ge k$, the projection
\[
  \pi_T:\C\longrightarrow\F_p^T,
  \qquad
  \pi_T(a)=a_T,
\]
is injective. Thus the coordinates $a_T$ determine
the entire codeword. When $|T|=k$, this projection
is bijective, because both $\C$ and $\F_p^T$
have $p^k$ elements.

Relabel the coordinates so that the elements of each block occupy
consecutive positions in the cyclic order
\(1,2,\ldots,n\) and $1$ comes after $n$. Coordinate permutations preserve the MDS property.

For each $r\in[n]$, let
\[
  I_r=\{r,r+1,\ldots,r+k-1\},
\]
where indices are interpreted modulo $n$.
Thus $I_r$ is a cyclic interval of $k$ coordinates.
Define
\[
  S_r=\{j\in[M]:B_j\cap I_r\ne\varnothing\},
  \qquad
  T_r=\bigcup_{j\in S_r}B_j.
\]
We include the whole block whenever the interval meets
it. Hence
\[
  I_r\subseteq T_r.
\]
Since $|I_r|=k$, the values on $I_r$ determine
the entire codeword, including all coordinates in
$T_r\setminus I_r$.

For a fixed block index \(j\), we now count the number of
starting positions \(r\) for which \(j\in S_r\).
Suppose that $B_j$ starts at position $s$ and has
length $m_j$. An interval of length $k$ meets this
block precisely when its starting position is one of
\[
  s-k+1,\ s-k+2,\ \ldots,\ s+m_j-1,
\]
with positions interpreted modulo $n$.
The number of positions in this list is
\[
  (s+m_j-1)-(s-k+1)+1
  =
  k+m_j-1.
\]
Since $k+m_j-1\le n$, these positions are distinct
modulo $n$.

Thus, setting
\[
  \nu_j=k+m_j-1,
\]
we have
\[
  |\{r\in[n]:j\in S_r\}|=\nu_j,
  \qquad
  q_j=\frac{n}{\nu_j}.
\]
For each \(r\in[n]\), define the nonnegative function
\[
  F_r(a)
  =
  \prod_{j\in S_r}g_j(a_{B_j})^{q_j},
  \qquad a\in\C.
\]
For a fixed \(j\), the factor \(g_j(a_{B_j})^{q_j}\)
appears in \(F_r(a)\) precisely when \(j\in S_r\).
This occurs for exactly
\(\nu_j\) values of \(r\). Therefore,
\[
\begin{aligned}
  \prod_{r=1}^n F_r(a)^{1/n}
  &=
  \prod_{j=1}^M
       g_j(a_{B_j})^{q_j\nu_j/n}\\
  &=
  \prod_{j=1}^M g_j(a_{B_j}),
\end{aligned}
\]
where the last equality follows from
$q_j\nu_j/n=1$.
Applying \cref{fact:holder} with $\Omega=\C$
and $R=n$, we obtain
\begin{equation}
  \label{eq:unequal-block-holder}
\begin{aligned}
  \sum_{a\in\C}\prod_{j=1}^M g_j(a_{B_j})
  &=
  \sum_{a\in\C}\prod_{r=1}^n F_r(a)^{1/n}
  \le
  \prod_{r=1}^n
  \left(\sum_{a\in\C}F_r(a)\right)^{1/n}.
\end{aligned}
\end{equation}

Fix $r$. Projection onto $I_r$ is bijective.
Thus, for every $u\in\F_p^{I_r}$, there is exactly
one codeword $a(u)\in\C$ satisfying
\[
  a(u)_{I_r}=u.
\]
Summing over codewords is therefore the same as summing
over all assignments $u$ to $I_r$.

Once $u$ is fixed, the values of $a(u)$ on the
remaining coordinates $T_r\setminus I_r$ are also
fixed. The function $F_r$ depends only on the
coordinates in $T_r$. Summing over all possible values of
the coordinates in $T_r\setminus I_r$, instead of
only the values specified by $a(u)$ adds only
nonnegative terms. Hence
\[
\begin{aligned}
  \sum_{a\in\C}F_r(a)
  &=
  \sum_{u\in\F_p^{I_r}}
       \prod_{j\in S_r}g_j(a(u)_{B_j})^{q_j}\\
  &\le
  \sum_{z\in\F_p^{T_r}}
       \prod_{j\in S_r}g_j(z_{B_j})^{q_j}\\
  &=
  \prod_{j\in S_r}
  \left(
    \sum_{u_j\in\F_p^{B_j}}g_j(u_j)^{q_j}
  \right).
\end{aligned}
\]
For each block, let
\[
  Z_j
  =
  \sum_{u\in\F_p^{B_j}}g_j(u)^{q_j}.
\]
Then, we have
\[
  \sum_{a\in\C}F_r(a)
  \le
  \prod_{j\in S_r}Z_j.
\]
Using \cref{eq:unequal-block-holder}, we get
\[
\begin{aligned}
  \sum_{a\in\C}\prod_{j=1}^M g_j(a_{B_j})
  &\le
  \prod_{r=1}^n
       \left(\prod_{j\in S_r}Z_j\right)^{1/n}\\
  &=
  \prod_{j=1}^M Z_j^{\nu_j/n}\\
  &=
  \prod_{j=1}^M Z_j^{1/q_j}.
\end{aligned}
\]
The first equality uses the fact that block $j$
belongs to exactly $\nu_j$ of the sets $S_r$;
the second uses $\nu_j/n=1/q_j$.

Substituting the definition of $Z_j$, we get
\[
  \sum_{a\in\C}\prod_{j=1}^M g_j(a_{B_j})
  \le
  \prod_{j=1}^M
  \left(
    \sum_{u\in\F_p^{B_j}}g_j(u)^{q_j}
  \right)^{1/q_j}.
\]
\end{proof}

\end{document}